\documentclass[12pt,draftcls,letterpaper,onecolumn]{IEEEtran}
\usepackage{setspace}

\usepackage{amsmath}
\usepackage{amsthm}
\usepackage{amssymb}
\usepackage{verbatim}

\usepackage{thmtools,thm-restate}

\usepackage{graphicx}
\usepackage{setspace}

\usepackage{tikz}
\usepackage{pgfplots}
\usepackage{stfloats}
\usepackage{subfigure}
\usepackage{epsfig}
\usepackage{ulem}

\usepackage{url}

\usepackage{tikz}
\usepackage{verbatim}
\usetikzlibrary{arrows,shapes}

\usetikzlibrary{shapes,snakes}
\usetikzlibrary{positioning,chains,calc}
\usepackage{relsize}
\usetikzlibrary{patterns}
\usepackage{colortbl}
\definecolor{kugray5}{RGB}{224,224,224}

\usepackage[linesnumbered,ruled,vlined]{algorithm2e}

\usepackage{pgfplots}
\pgfplotsset{compat=newest}

\usepgfplotslibrary{patchplots}

\theoremstyle{plain}
\newtheorem{lemma}{Lemma}
\newtheorem{theorem}{Theorem}
\newtheorem{proposition}[theorem]{Proposition}
\newtheorem{definition}{Definition}

\theoremstyle{definition}
\newtheorem{example}{Example}

\newcommand{\argmax}{\operatornamewithlimits{argmax}}

\allowdisplaybreaks[4]
\IEEEoverridecommandlockouts

\begin{document}

\newlength{\figurewidth}\setlength{\figurewidth}{0.6\columnwidth}

\title{Controlled Matching Game for Resource Allocation and User Association in WLANs
}

\newcounter{one}
\setcounter{one}{1}
\newcounter{two}
\setcounter{two}{2}
\newcounter{three}
\setcounter{three}{3}
\newcounter{four}
\setcounter{four}{4}

\author{	 Mikael Touati,
		Rachid El-Azouzi,
		Marceau Coupechoux,
		 Eitan Altman and Jean-Marc Kelif
\IEEEcompsocitemizethanks{
Mikael Touati and Jean Marc Kelif are with Orange Lab, France,
E-mail: \{mikael.touati, jeanmarc.kelif\}@orange.com,
Rachid El-Azouzi is  with CERI/LIA, University of Avignon,  Avignon, France,
E-mail: rachid.elazouzi@univ-avignon.fr,
Mikael Touati and Marceau Coupechoux are with LTCI, CNRS, Telecom ParisTech, University Paris-Saclay, France,
E-mail:  \{mikael.touati, marceau.coupechoux\}@telecom-paristech.fr,
Eitan Altman is with INRIA Sophia Antipolis, France,
E-mail: {eitan.altman@inria.fr}.
}
}


\addtolength{\floatsep}{-\baselineskip}
\addtolength{\dblfloatsep}{-\baselineskip}
\addtolength{\textfloatsep}{-\baselineskip}
\addtolength{\dbltextfloatsep}{-\baselineskip}
\addtolength{\abovedisplayskip}{-1ex}
\addtolength{\belowdisplayskip}{-1ex}
\addtolength{\abovedisplayshortskip}{-1ex}
\addtolength{\belowdisplayshortskip}{-1ex}

\IEEEtitleabstractindextext{%
\begin{abstract}
In multi-rate IEEE 802.11 WLANs, the traditional user association based on the strongest received signal and the well known anomaly of the MAC protocol can lead to overloaded Access Points (APs), and poor or heterogeneous performance.  
Our goal is to propose an alternative game-theoretic approach for association. We model the joint resource allocation and user association as a matching game with complementarities and peer effects consisting of selfish players solely interested in their individual throughputs. Using recent game-theoretic results we first show that various resource sharing protocols actually fall in the scope of the set of stability-inducing resource allocation schemes. 
 The game makes an extensive use of the Nash bargaining and some of its related properties that allow to control the incentives of the players. We show that the proposed mechanism can greatly improve the efficiency of 802.11 with heterogeneous nodes and reduce the negative impact of peer effects such as its MAC anomaly. The mechanism can be implemented as a virtual connectivity management layer to achieve efficient APs-user associations without modification of the MAC layer.   
\end{abstract}

\begin{IEEEkeywords}
Game theory, matching games, coalition games, stability, control, 802.11.
\end{IEEEkeywords}}

\maketitle
\IEEEdisplaynontitleabstractindextext

\section{Introduction} \label{sec:intro}
The IEEE 802.11 based wireless local area networks (WLANs) have attained a huge popularity in dense areas as public places, universities and city centers. In such environments, devices have the possibility to use many Access Points (APs) and usually a device selects an AP with the highest received Radio Signal Strength Indicator (best-RSSI association scheme). In this context, the performance of IEEE 802.11 may be penalized by the so called 802.11 {\it anomaly} and by an imbalance in AP loads (congestion). 
Moreover, some APs may be overloaded while others are underutilized because of the association rule. 

In this paper, we consider a fully distributed IEEE 802.11 network, in which selfish mobile users and APs look for the associations maximizing their individual throughputs. We analyze this scenario using matching game theory and develop a unified analysis of the joint mobile user association and resource allocation problem for the reduction of the anomaly and for load balancing in IEEE 802.11 WLANs. In a network characterized by a {\it state of nature} (user locations, channel conditions, physical data rates), composed of a set $\mathcal{W}$ of mobile users and a set $\mathcal{F}$ of APs, the user association problem consists in finding a mapping $\mu$ that associates every mobile user to an AP. We call the set formed by an AP and its associated mobile users a {\it cell}, or a {\it coalition} in the game framework. The set of coalitions induced by $\mu$ is called a matching or a {\it structure} (partition of the players in coalitions). Once mobile user association has been performed, a resource allocation scheme (also called a {\it sharing rule} in this paper) 
allocates radio resources of a cell to the associated mobile users. 

%

This matching game is characterized by {\it complementarities} in the sense that APs have preferences over groups of mobile users and {\it peer effects} in the sense that mobile users care who their peers are in a cell and thus emit preferences also over groups of mobiles users. 
Indeed, by definition of DCF implementation of the IEEE 802.11 protocol, a users' throughput does not only depend on its physical data rate but also on the coalition size and composition. We are thus facing the classical association problem with the additional property that the players (mobile users and APs) are selfish and solely interested in the association maximizing their own throughput. The following questions are raised: does there exist associations (or matchings) in which no subset of players prefer deviating and associate with each others, i.e., are there stable associations? Do these associations always exist? Is there unicity? How to reach these equilibria in a decentralized way? Finally, how to provide the players the incentive to make the system converge to another association point with interesting properties in terms of load balancing?

Assuming that players associate solely w.r.t their individual throughput many mobile users may remain unassociated since every AP has the incentive to associate with a single mobile user having the best data rate. We call this problem the {\it unemployment problem}. 
To  counter this side effect and provide the nodes the incentives to associate with each others, we design a decentralized three steps mechanism to control the set of the stable matchings. 
In the first step, the APs share the load.
In the second step, the coalition game is controlled to provide the incentives to enforce the load balancing. 
In the third step, players play the {\it controlled} coalition game with individual payoffs obtained from a Nash Bargaining (NB)-based sharing rule. 
This sharing rule is interesting because it generalizes equal sharing, but also other proposals in the literature such as proportional fairness.
Under  some assumptions, the NB-based sharing rule guarantees that the set of stable structures is non empty in all states of nature. 
The control of the game is designed so as to provide the players the incentives to respect the  objective of the load balancing (during the first step). The control is based on the notion of Fear of Ruin (FoR) introduced in~\cite{Aum1977}.
The equilibrium point of the third step is obtained by a decentralized algorithm that results in a core stable  matching (or structure). 
We propose here a modified version of the Deferred Acceptance Algorithm (DAA), called Backward Deferred Acceptance Algorithm (BDAA), for matching games with complementarities {and  peer effects. Similarly to the DAA, the complexity of the BDAA is polynomial. 
We show through numerical simulations that our mechanism not only ensures that a stable matching will form but is also a way to reduce the impact of the WiFi anomaly. 
In fact, the equilibrium association relies relies on the agents' incentives to counter the side effects induced by the protocol. 
Moreover, this mechanism allows us to exploit the overlapping of APs as an opportunity to reduce the anomaly of 802.11 rather than an obstacle.  

\subsection{ Related Work}\label{subsec:relatedworks}
IEEE 802.11 (WiFi) anomaly is a well documented issue in the literature, see e.g.~\cite{Tan2004, Bab2005, Ban2007}. The first idea to improve the overall performance of a single cell system is to modify the MAC so as to achieve a {\it time-based fairness} \cite{Tan2004, Bab2005}. Authors of \cite{Tan2004} propose a leaky-bucket like approach. Banchs et al. \cite{Ban2007} achieve {\it proportional fairness} by adjusting the transmission length or the contention window parameters of the stations depending on their physical data rate. Throughput based fairness, time based fairness and proportional fairness resource allocation schemes are sharing rules that can be obtained from a Nash bargaining as we will see later on. 


In a multiple cell WLAN network, mobile user-AP association plays a crucial role for improving the network performance and can be seen as a mean to mitigate the WiFi anomaly without modifying the MAC layer. The maximum RSSI association approach, though very simple, may cause an imbalanced traffic load among APs, so that many devices can connect to few APs and obtain low throughput, while few of them benefit from the remaining radio resource. Kumar et al. \cite{Kum2005} investigate the problem of maximizing the sum of logarithms of the throughputs. Bejerano et al. \cite{Bej2007} formulate a mobile user-AP association problem guaranteeing a max-min fair bandwidth allocation for mobile user. This problem is shown to be NP-hard and constant-factor approximation algorithms are proposed. Li et al. \cite{Li2011}.

 Arguing for ease of implementation, scalability and robustness, several papers have proposed decentralized heuristics to solve this issue, see e.g. \cite{Kor2005, Gon2008, Bon2009}. Reference \cite{Kor2005} proposes to enhance the basic RSSI scheme by an estimation of the Signal to Interference plus Noise Ratio (SINR) on both the uplink and the downlink. Bonald et al. in \cite{Bon2009} show how performance strongly depends on the frequency assignment to APs and propose to use both data rate and MAC throughput in a combined metric to select the AP. 
 Several papers have approached the problem using game theory based  on individual MAC throughput.  Due to the WiFi anomaly, this is not a classical {\it crowding game} in the sense that the mobile user achieved throughput is not necessary a monotonically decreasing function of the number of attached devices, as it can be the case in cellular networks \cite{Che2008,  Xu2010}. Compared to proposed decentralized approaches, we do not intend to optimize some network wide objective function, but rather to study the equilibria resulting from selfish behaviors. Compared to other game-theoretic approaches, we consider a fully distributed scenario, in which APs are also players able to accept or reject mobile users. This requires the study of the core stability, a notion stronger than the classical Nash Equilibrium. Moreover, there is a need in understanding the fundamental interactions between mobile user association and resource allocation in the presence of complementarities and peer effects. 

In this paper, we tackle the mobile user-AP association problem using the framework of matching games with individual selfish players. This framework provides powerful tools for analyzing the stability of associations resulting from decentralized mechanisms. Matching games~\cite{Rot1990} is a field of game theory that have proved to be successful in explaining achievements and failures of matching and allocation mechanisms in decentralized markets. Gale and Shapley published one of the earliest and probably most successful paper on the subject \cite{Gale1962} and solved the stable marriage and college admissions association problem with a polynomial time algorithm called DAA. 

Some very recent papers in the field of wireless networks have exploited the theoretical results and practical methods of matching games \cite{Pan2013, Ham2014, Saa2014}, although none has considered the WLAN association problem and its related WiFi anomaly. Authors of \cite{Pan2013} address the problem of downlink association in wireless small-cell networks with device context awareness. The relationship between resource allocation and stability is not investigated and APs are not allowed to reject users. Hamidouche et al. in \cite{Ham2014} tackle the problem of video caching in small-cell networks. They propose an algorithm that results in a many-to-many pairwise stable matching. Preferences emitted by servers exhibit complementarities between videos and vice versa. Nevertheless, the model doesn't take into account peer effects within each group. Reference \cite{Saa2014} addresses the problem of uplink user association in heterogeneous wireless networks. Invoking a high complexity, complementarities are taken into account by a transfer mechanism that results in a Nash-stable matching, a concept weaker than pairwise stability or core stability.



\subsection{Contributions}
Our contributions can be summarized as follows:\\
{$\bullet$} We provide a matching game-theoretic unified approach of mobile user association and resource allocation in IEEE 802.11 WLANs in the presence of complementarities and peer effects. 
The results of the paper highlight the importance of the Nash bargaining in wireless networks as a stability inducer but also as a convenient and easy-to-use tool at several levels of network resource management.\\
To the best of our knowledge, this is the first game-theoretic modeling of the IEEE 802.11 protocol  covering such a number of resource allocation mechanisms proposed in the literature.\\
{$\bullet$}   We use existing theoretical results to show that if the scheduling and/or the MAC protocol result from a Nash bargaining then there exist stable mobile user associations, whatever the user data rates or locations.\\
{$\bullet$}  In order to {\it control} the matching game, we design a three steps mechanism, which includes 1) a generic load balancing, 2) a control step i.e. the modifications to be applied to the worths of the coalitions in order to provide the agents the incentives to enforce the result of the load balancing, 3) a coalition game with resource allocation defined as a Nash bargaining over the resource to be allocated and a stable matching algorithm with players' preferences induced by the resource allocation. 
This three steps mechanism tackles the so called unemployment problem, that would have left mobile users aside from the association otherwise. We show through numerical examples that our mechanism achieves good performance compared to the global optimum solution. We also show how the mechanism can be used to efficiently share the load between APs.\\
To the best of our knowledge, such a mechanism is absent from both the game theoretic and wireless networks based on matching games literature.
 
{$\bullet$}   We show that our BDAA can be efficiently used to find a stable many-to-one matching in a coalition game with complementarities and peer-effects.
The algorithm has a polynomial complexity.

The mechanism has been originally proposed in an extended abstract~\cite{Tou2015c}. 
In this paper, we provide a complete description, show the mathematical results and corresponding proofs. 
We furthermore extend the results to a generic load balancing scheme and to the Nash bargaining-based sharing rules (resource allocation schemes). 
The particular case of equal sharing has already been assessed in~\cite{Tou2015b}.
BDAA has been originally proposed in a short paper~\cite{Tou2015} and we provide here proofs of convergence to a core stable structure and of polynomial complexity.
\\

The rest of the paper is organized as follows. In Section~\ref{sec:systemmodel}, we define the system model. In Section~\ref{sec:matching}, we  formulate the IEEE 802.11 WLANs resource allocation and decentralized association problem. 
In Section ~\ref{sec:existencestable}, we show that there exist stable coalition structures under certain conditions whatever the individual data rates. Section~\ref{sec:mechanism} presents our three steps mechanism. Section~\ref{sec:wlan} shows numerical results. Section~\ref{sec:Conclusion} concludes the paper and provides perspectives. 

\section{System model}\label{sec:systemmodel}
We summarize in Table \ref{table:notations} the notations used in this paper. We use both game-theoretic definitions and their networking interpretation. Throughout the paper, they are used indifferently. 
\begin{table}[t]
\footnotesize
\begin{center}
\begin{tabular}{llll|}
\hline
$|set|$& cardinality of the set $set$& $\mathcal{N}$& set of players (mobile users and APs)\\
$\mathcal{W}$& set of mobile users &$\mathcal{F}$& set of Access Points (APs)\\
$\mathcal{C}$& set of coalitions (cells)& $\mathcal{C}_{f}$& set of coalitions containing AP $f\in\mathcal{F}$\\
$C$& coalition (cell)& $\mu$& matching (AP-mobile user association)\\
$\Theta$& set of feasible data rates & $\theta_{wf}$& data rate between $w$ and $f$\\
$r_{i,C}$& throughput of node (user or AP) $i$ in cell $C$& $\alpha_{i,C}$& proportion of resources (time, frequency) of $i$ in cell $C$\\
$D$& sharing rule (resource allocation scheme) & $v(C)$& worth of coalition $C$\\
$s_{i,C}$& payoff of player $i$ in coalition $C$& $u_{i}(.)$& utility function of player $i$\\
$q_{i}$& quota of player $i$& $\chi_{C}$& fear-of-ruin of coalition $C$\\ 
$P(i)$& preferences list of player $i$ over individuals & $P^{\#}(i)$& preferences list of player $i$ over groups\\
\hline
\end{tabular}
\end{center}
\caption{Notations}
\label{table:notations}
\end{table}%
Let define the set of players (nodes) $\mathcal{N}$ of cardinality $N$ as the union of the disjoint sets of mobile users $\mathcal{W}$ of cardinality $W$ and APs $\mathcal{F}$ of cardinality $F$. 
As in \cite{Kum2005}, we assume an interference-free model. It is assumed that the AP placement and channel allocation are such that the interference between cochannel APs can be ignored.
In game-theoretic terms, this implies that there are no externalities.
The mobile user association is a mapping $\mu$ that associates every mobile user to an AP and every AP to a subset of mobile users. 
 
The IEEE 802.11 standard MAC protocol has been set up to enable any node in $\mathcal{N}$ to access a common medium in order to transmit its packets. The physical data rate between a transmitter and a receiver depends on their respective locations and on the channel conditions. For each mobile user $i\in\mathcal{W}$, let $\theta_{if}$ be the (physical) data rate with an AP $f$ where $\theta_{if}\in \Theta=\{\theta^{1},\ldots,\theta^{m}\}$, a finite set of finite rates resulting from the finite set of Modulation and Coding Schemes. 
If $i$ is not within the coverage of $f$ then $\theta_{if}=0$. Given an association $\mu$, let  $\boldsymbol{\theta}_C=(\theta_{wf})_{(f,w)\in (C\cap \mathcal{F})\times (C\cap \mathcal{W})}$ denote the data rate vector of mobiles users in cell $C$ served by AP $f$. Let $\mathbf{n_{C}}$ be the normalized composition vector of $C$, whose $k$-th component is the proportion of users in $C$ with data rate $\theta_{k}\in \Theta$. Note that an AP is defined in the model as a player with the additional property of having maximum data rate on the downlink.
Within each cell, a resource allocation scheme (e.g. induced by the CSMA/CA MAC protocol) may be  formalized as a sharing rule over the  {resource to be shared in the cell.
This  resource may be the total cell throughput (as considered in the saturated regime) or the amount of radio resources in time or frequency in the general case. More precisely, a sharing rule is a set of functions $D = (D_{i,C})_{C\in\mathcal{C},i\in C}$, where $D_{i,C}$ allocates a part of the resource of $C$ to user $i\in C$. Equal sharing, proportional fairness, $\alpha$-fairness are examples of sharing rules. 

Assuming the IEEE 802.11 MAC protocol and the saturated regime, the overall cell resource of cell $C$ is defined as the total throughput. It is a function of the composition vector $\mathbf{n_{C}}$ and of the cardinality $|C|$. We denote $r_{i,C}$ the throughput obtained by user $i$ in cell $C$. From the  game theoretic point of view, $r_{i,C}$ is understood as $i$'s share of the worth of coalition $C$ denoted $v(C)$. The function $v:\mathcal{C}\rightarrow\mathbb{R}$ is called the characteristic function of the coalition game and maps any coalition $C\in \mathcal{C}$ to its worth $v(C)$.
Other MAC protocols and regime can however be modeled by this approach. 
For example time-based fairness proposed in the literature to solve the WiFi anomaly results from the sharing of the time resource. In this case, a user $i$ gets a proportion $\alpha_{i,C}$ of the time resource, which induces a throughout of $\alpha_{i,C}\theta_{if}$, where $f$ is the AP of $C$. It can be shown that time-based fairness results in a proportional fairness in throughputs. 



\section{Matching Games Formulation} \label{sec:matching}
\subsection{Matching Games for Mobile User Association}
In this paper, the mobile user association is modeled as a matching game (in the class of coalition games).  The matching theory relies on the existence of individual order relations $\{\succeq_{i}\}_{i\in\mathcal{N}}$, called preferences, giving the player's ordinal ranking
\footnote{In this paper, we use the Individually Rational Coalition Lists (IRCLs) to represent preferences. It can indeed easily be shown that other representations (additively separable preferences, B-preferences, W-preferences) are not adapted to our problem, see~\cite{Kei2011} for more details.
}
of alternative choices. 
As an example, $w_{1}\preceq_{f_{1}}[w_{2},w_{3}]\preceq_{f_{1}} w_{4}$ indicates that the AP $f_{1}$ prefers to be associated to mobile user $w_{4}$ to any other mobile user, is indifferent between $w_{2}$ and $w_{3}$, and prefers  to be associated to mobile user $w_{2}$ or $w_{3}$ rather than to be associated to $w_{1}$. Following the notations of Roth and Sotomayor in \cite{Rot1990}, let us denote $\textbf{P}$ the set of preference lists $\textbf{P}=(P_{w_{1}},\dots,P_{w_{W}},P_{f_{1}},\ldots,P_{f_{F}})$. 
\begin{definition}[Many-to-one bi-partite matching \cite{Rot1990}]
	A matching $\mu$ is a function from the set $\mathcal{W}\cup\mathcal{F}$ into the set of all subsets of $\mathcal{W}\cup\mathcal{F}$ such that:\\
	(i) $|\mu(w)|=1$ for every mobile user $w\in\mathcal{W}$ and $\mu(w)=w$ if $\mu(w)\not\in\mathcal{F}$;\\
	(ii) $|\mu(f)|\leq q_{f}$ for every AP $f\in\mathcal{F}$ ($\mu(f)=\emptyset$ if $f$ isn't matched to any mobile user in $\mathcal{W}$);\\
	(iii) $\mu(w)=f$ if and only if $w$ is in $\mu(f)$.
\end{definition}
Condition (i) of the above definition means that a mobile user can be associated to at most one AP and that it is by convention associated to itself if it is not associated to any AP. Condition (ii) states that an AP $f$ cannot be associated to more than $q_f$ mobile users. Condition (iii) means that if a mobile user $w$ is associated to an AP $f$ then the reverse is also true. In this definition, $q_{f}\in{\mathbb{N}^{*}}$ is called the {\it quota} of AP $f$ and it gives the maximum number of mobile users the AP $f$ can be associated to.

From now on, we focus on many-to-one matchings.
In this setting, stability plays the role of equilibrium solution. 
In this paper, we particularly have an interest in the pairwise and core stabilities. 
For more details we refer the reader to the reference book \cite{Rot1990}. We say that a matching $\mu$ is {\it blocked by a player} if this player prefers to be unmatched rather than being matched by $\mu$. We say that it is {\it blocked by a pair} if there exists a pair of unmatched players that prefer to be matched together.
\begin{definition}[Pairwise sability \cite{Rot1990}]
A matching $\mu$ is {\bf pairwise stable} if it is not blocked by any player or any pair of players. The set of pairwise stable matchings is denoted $S(\mathbf{P})$. 
\end{definition}
\begin{definition}[Domination \cite{Rot1990}]
 	A matching $\mu'$ dominates another matching $\mu$ via a coalition $C$ contained in $\mathcal{W}\cup\mathcal{F}$ if for all mobile users $w$ and APs $f$ in $C$, (i) if $f' = \mu'(w)$ then $f'\in C$, and if $w'\in\mu'(f)$ then $w'\in C$; and (ii) $\mu'(w)\succ_{w} \mu(w)$ and $\mu'(f)\succ_{f}\mu(f)$.
\end{definition}
\begin{definition}[Weak Domination \cite{Rot1990}]
 	A matching $\mu'$ weakly dominates another matching $\mu$ via a coalition $C$ contained in $\mathcal{W}\cup\mathcal{F}$ if for all mobile users $w$ and APs $f$ in $C$, (i) if $f' = \mu'(w)$ then $f'\in C$, and if $w'\in\mu'(f)$ then $w'\in C$; and (ii) $\mu'(w)\succeq_{w} \mu(w)$ and $\mu'(f)\succeq_{f}\mu(f)$; and (iii) $\mu'(w)\succ_{w}\mu(w)$ for some $w$ in $C$, or $\mu'(f)\succ_{f} \mu(f)$ for some $f$ in $C$.
\end{definition}
\begin{definition}[Cores of the game \cite{Rot1990}]
	The core $C(\mathbf{P})$ (resp. the core defined by weak domination $C_{W}(\mathbf{P})$) of the matching game is the set of matchings that are not dominated (resp. weakly dominated) by any other matching. 
\end{definition}
In the general case, the core of the game $C(\mathbf{P})$ contains $C_{W}(\mathbf{P})$. When the game does not exhibit complementarities or peer effects, it is sufficient for its description that the preferences are emitted over individuals only. In the presence of complementarities or peer effects, players in the same coalition (i.e. the set of mobile users matched to the same AP) have an influence on each others. In such a case, the preferences need to be emitted over subsets of players and are denoted $P^{\#}$. 

In the classical case of matchings with complementarities, the preference lists are of the form $\mathbf{P} = (P_{w_{1}},\ldots,P_{w_{W}},P^{\#}_{f_{1}},\ldots,P^{\#}_{f_{F}})$, i.e., preferences over groups are emitted only by the APs (see the firms and workers problem in~\cite{Rot1990}). Moreover, it may happen that the preferences over groups may be {\it responsive} to the individual preferences in the sense that they are aligned with the individual preferences in the preferences over groups differing from at most one player. The preferences over groups may also satisfy the substitutability property. The substitutability of the preferences of a player rules out the possibility that this player considers others as complements.
\begin{definition}[Responsive preferences \cite{Rot1990}]
	The preferences relation $P^{\#}(i)$ of player $i$ over sets players is responsive to the preferences $P(i)$ over individual players if, whenever $\mu'(i) = \mu(i)\cup\{k\}\backslash\{l\}$ for $l$ in $\mu(i)$ and $k$ not in $\mu(i)$, then $i$ prefers $\mu'(i)$ to $\mu(i)$ (under $P^{\#}(i)$) if and only if $i$ prefers $k$ to $l$ (under $P(i)$).
\end{definition}
Before defining the substitutability property of preferences, we need to introduce the choice function $Ch_{i}$ of a player $i$. Given any subset $C$ of players, $Ch_{i}(S)$ is called the choice set of $i$ in $S$.
It gives the subset of players in $C$ that player $i$ most prefers. 

\begin{definition}[Substitutable preferences \cite{Rot1990}]
	A player $i$'s ($i\in\mathcal{W}\cup\mathcal{F}$) preferences over sets of players has the property of substitutability if, for any set $S$ that contains players $k$ and $l$, if $k$ is in $Ch_{i}(S)$ then $k$ is in $Ch_{i}(S\backslash{l})$.
\end{definition}
Considering preference lists of the form $\mathbf{P} = (P_{w_{1}},\ldots,P_{w_{W}},P^{\#}_{f_{1}},\ldots,P^{\#}_{f_{F}})$ and assuming either responsive or substitutable strict preferences, we have the result that $C_{W}(\mathbf{P})$ equals $S(\mathbf{P})$. Any many-to-one matching problem with these properties has an equivalent one-to-one matching problem, which can be solved by considering preferences over individuals only. The set of pairwise stable matching is non-empty.

If the preferences are neither responsive nor substitutable, the equality $S(\mathbf{P})=C_{W}(\mathbf{P})$ does not hold in general and the sets of pairwise, weak core and core stable matchings may be empty. An additional difficulty appears if the preferences over groups have to be considered on the mobile users side, i.e., if we have preference lists of the form $\mathbf{P} = (P^{\#}_{w_{1}},\ldots,P^{\#}_{w_{W}},P^{\#}_{f_{1}},\ldots,P^{\#}_{f_{F}})$. Complementarities and peer effect may arise in both sides of the matching. The user association problem in IEEE 802.11 WLANs falls in this category because the performance of any mobile user in a coalition may depend on the other mobiles in the coalition. To break the indifference, we use the following rule: a mobile user prefers a coalition with AP with the lowest index and an AP prefers coalitions in lexicographic order of users indices.


To see that preferences may not be responsive, consider an example with only uplink communications, two APs $f_1$ and $f_2$ and three mobile users $w_1, w_2, w_3$ such that $\theta_{11}=300$~Mbps, $\theta_{12}=\theta_{22}=54$~Mbps, $\theta_{21}=\theta_{32}=1$~Mbps. Assuming saturated regime and equal packet size, we can show that $P^{\#}(w_1) = f_{1}\succ f_{2} \succ \{w_{3};f_{1}\} \succ \{w_{2};f_{2}\} \succ \{w_{2};f_{1}\} \succ \{w_{3};f_{2}\}$, which is not responsive. In this example, we also see that substitutability is not even defined since every choice set is reduced to a singleton. After the game has been controlled according the proposed mechanism, preferences of $w_1$ can be modified as follows: $P^{\#}(w_1) =  \{w_{3};f_{1}\} \succ \{w_{2};f_{2}\} \succ \{w_{2};f_{1}\} \succ \{w_{3};f_{2}\} \succ f_{1}\succ f_{2}$. Considering $S=\{w_2,w_3;f_1,f_2\}$, we have $Ch_{w_1}(S)=\{w_3;f_1\}$, while $Ch_{w_1}(S\backslash w_3)=\{w_2;f_2\}$. Preferences are thus not substitutable. 

This general many-to-one matching problem has algorithmically been assed by Echenique and Yenmen in~\cite{Ech2007}  who propose a fixed-point formulation and an algorithm to enumerate the set of stable matchings. 
It is known, that there is no guarantee that this set is non empty if the individual preferences over groups are not of a particular form. 
The problem of complementarities and peer effects in matchings has been analytically tackled by Pycia in \cite{Pyc2012}. 
Nevertheless, no result have been derived concerning the decentralized control of core stable structures and no decentralized algorithm with a limited amount of information and reduced lists of preferences for the mobiles have been derived. 

\subsection{Sharing Rules and Matching Game Formulation}\label{subsec:SRMGF}
We now assume that a player $i$ in a given coalition $C$ obtains a {\it payoff} $s_{i,C}$, which is evaluated (or perceived) by it through a {\it utility function} $u_i:\mathbb{R}\rightarrow\mathbb{R}$. In this paper, we assume that functions $u_i$ are positive, concave (thus log-concave), increasing and differentiable. In such a case, the individual preferences are induced by the player's utilities of these payoffs. 
We extend our model to the framework of finite coalition games in characteristic form $\Gamma=(\mathcal{N};v)$, where  $v$ is a function mapping any coalition to its worth in $\mathbb{R}^{+}$. 
By definition of the characteristic function $v(\emptyset) = 0$.
In this paper we do not assume a particular form of the characteristic function $v$ (e.g. super-additivity\footnote{$\forall C,C, \; v(C\cup C')\geq v(C) + v(C')\; \text{ if } C\cap C' = \emptyset$}). 
An even particular case of coalition games in characteristic form concerns games with an exogenous sharing rule $\Gamma=(\mathcal{N};v;E^{N};D)$, where $E^{N}$ is the set of all payoff vectors and $D$ is a sharing rule. 

\begin{definition}[Sharing Rule] A sharing rule is a collection of functions $D_{i,C}:\mathbb{R}^{+}\rightarrow\mathbb{R}^{+}$, one for each coalition $C$ and each of its members $i\in C$, that maps the worth $v(C)$ of $C$ into the share of output obtained by player $i$. We denote the sharing rule given by functions $D_{i,C}$ as $D=(D_{i,C})_{C\in\mathcal{C},i\in C}$.
\end{definition}
From this definition, the payoff of user $i$ in coalition $C$ is given by $s_{i,C} = D_{i,C}\circ v(C)$ and his utility of this payoff is given by $u_{i}(s_{i,C})$.
We  can now formulate the IEEE 802.11 joint user association and resource allocation problem as a matching game. 
\begin{definition}[Resource Allocation and User Association Game]
Using the above notations, the resource allocation and users association game is defined as a $N$-player many-to-one matching game in characteristic form with sharing rule $D$ and rates $\boldsymbol{\theta}=\{\theta_{wf}\}_{(w,f)\in\mathcal{W}\times\mathcal{F}}$: $\Gamma=(\mathcal{W}\cup\mathcal{F},v,\mathbb{R}^{+N},D,\boldsymbol{\theta})$. Each pair of players of the form $(w,f)\in\mathcal{W}\cup\mathcal{F}$ is endowed with a rate $\theta_{wf}$ from the rates space $\Theta=\{\theta^{1},\ldots,\theta^{m}\}$.
For this game, we define the set of possible coalitions $\mathcal{C}$:
\begin{equation}
\mathcal{C}=\{\{f\} \cup J,\;f\in  {\mathcal F},\; J \subseteq  \mathcal{W}, |J| \leq  q_f \} \cup \{\{w\}, \; w\in \mathcal{W}\}. \label{eq:coalitionC}
\end{equation}
\end{definition}
Note that for IEEE 802.11 MAC protocol and for the saturated regime, $s_{i,C}\triangleq r_{i,C}$. For other time sharing MAC approaches, $s_{i,C}\triangleq \alpha_{i,C}$.

\section{Existence of Core Stable Structures}\label{sec:existencestable}
%

\subsection{Background on Nash Bargaining}\label{subsec:BNB}
The analytical theory of bargaining has mainly been developed on the concept introduced by J.F. Nash in~\cite{Nash1950a} and~\cite{Nash1950b} for the two person game and by Harsanyi in~\cite{Har1963} for the N-person game. The bargaining is developed as a cooperative game where the set of acceptable (feasible) individual payoffs results from the set of mutual agreements among the players involved. In this paper, we use the Nash bargaining to model the resource allocation scheme (see Section \ref{sec:mechanism}). 

Let $B\subset\mathbb{R}^{N}$ be the convex compact subset of jointly achievable utility points. Let $\mathbf{t} = (t_{1},\ldots,t_{N})$ be the fixed threat vector, an exogeneous parameter that may or not result from a threat game. The bargaining problem consists in looking for a payoff vector $(u_1,...,u_N)$ in $B$ satisfying five axioms: (i) Strong Efficiency, (ii) Individual Rationality, (iii) Scale Invariance, (iv) Independence of Irrelevant Alternatives and, (v) Symmetry.\\
We have the following result (see~\cite{Tou2006} and references therein):
\begin{theorem}\label{the:solutionnashproduitutilities}
	Let the utility functions $u_{i}$ be concave, upper-bounded and defined on $X$, a convex and compact subset of $\mathbb{R}^{n}$. Let $X_{0}$ be the set of payoffs s.t. $X_{0} = \{\mathbf{s}\in X|\forall i, u_{i}(s_i)\geq t_{i}\}$ and $J$ be the set of users s.t., $J = \{j = \{1,\ldots,N\}|\exists \mathbf{s}\in X, s.t.\; u_{j}(s_j)>t_{j}\}$. Assume that $\{u_{j}\}_{j\in J}$ are injective. Then there exists a unique Nash bargaining problem as well as a unique Nash bargaining solution $\mathbf{s}$ that verifies $u_{j}(s_j)>t_{j}$, $j\in J$, and is the unique solution of the problem:
	\begin{equation}
		\max_{\mathbf{s}\in X_0} \prod_{j\in J}(u_{j}(s_j)-t_{j}).
	\end{equation}
\end{theorem}
An important result about the Nash bargaining solution is that it achieves a generalized proportional fairness which includes as a special case the well-known and commonly-used proportional fairness. The proportional fairness is achieved in the utility space with a null threat vector. Nevertheless, if the players' utility functions are linear in their payoffs and the threats vector is null, the proportional fairness is achieved in the payoff space. In other words, the payoff vector induced by the Nash Bargaining over the coalition throughput is proportional fair. This makes game-theory and in particular the bargaining problem of fundamental importance in networks.
In Appendix~\ref{app:80211-NB}, we show that the resource allocation induced by the MAC protocol of IEEE 802.11 can be modeled as the result of a Nash bargaining. 

\subsection{On the Existence of Core Stable Structures}
In this section, we show the existence of stable coalition structures (users-AP association) when preferences are obtained under some regularity conditions over the set of coalitions and some assumptions over the monotonicity of the sharing rules. 
There exists a stable structure of coalitions whatever the state of nature $\boldsymbol{\theta}$ if and only if the sharing rules may be formulated as arising from the maximization of the product of increasing, differentiable and strictly log-concave individual utility functions in all coalitions. 

\begin{proposition}[\cite{Pyc2012}, Corollary 2] \label{prop:stability}
If the set of coalitions $\mathcal{C}$ is such that $q_{f}\in\{2,\ldots,W-1\}$ and $F\geq 2$, and if the sharing rule $D$ is such that all functions $D_{i,C}$ are striclty increasing, continuous and $\lim_{y\rightarrow +\infty}D_{i,C}(y)=+\infty$, then there is a stable coalition structure for each preference profile induced by the sharing rule $D$ iff there exist increasing, differentiable, and strictly log-concave functions $u_{i}:\mathbb{R}^{+}\rightarrow\mathbb{R}^{+}, i\in\mathcal{N}$, such that $\frac{u_{i}(0)}{u'_{i}(0)} = 0$ and
	\begin{equation} \label{eq:NBresourcealloc}
	(D_{i,C}\circ v(C))_{i\in C} = \argmax\limits_{\mathbf{s}_C\in B_C}\prod\limits_{i\in C}u_{i}(s_{i,C}),
	\end{equation}
where $C\in\mathcal{C}$ and $B_C=\{\mathbf{s}_C=(s_{i,C})_{i\in C}|\sum_{i\in C}s_{i,C}\leq v(C)\}$.
\end{proposition}

According to Proposition~\ref{prop:stability}, there are two conditions to be satisfied by the coalitions of our resource allocation and user association game: (i) consider scenarios with at least two APs (which is reasonable when talking about load balancing) and (ii) every AP is supposed to be able to serve at least two users and should not be able to serve the whole set of users. For more details, see Appendix~\ref{app:regularity}.

Proposition~\ref{prop:stability} ensures that there exists a stable coalition structure as soon as resource allocation results from a Nash bargaining.  The equal sharing resulting from CSMA/CA MAC protocol in saturated regime, single-flow per device and equal packet length} is obtained by considering $s_{i,C} = r_{i,C}$ and the identity function for $u_i$. 
 The players' throughputs in the general saturated regime with multiple flows and heterogenous packet length is obtained by taking $s_{i,C} = r_{i,C}$ and utility functions as shown in Appendix \ref{app:80211-NB}. 
Time-based fairness is obtained by setting $s_{i,C} = \alpha_{i,C}$ and the identity function for $u_i$. It results in turn in proportional fairness in terms of individual throughputs.

Without controlling the coalition game, some core stable structures can be formed in some  cases. For example, in CSMA/CA under saturated regime, the cell throughput is increasing with the individual physical data rates and individual throughputs $r_{i,C}$ are sub-additive, i.e., decreasing with the addition of users. Assuming that the payoff is the individual throughput, i.e., $s_{i,C} = r_{i,C}$, then each player has the incentive to form the lowest cardinality coalition with highest composition vector. 
In this case, the unique stable structure is a one-to-one matching, in which APs are associated to their best mobile user. 
This will further be mentioned in the name of the \textit{unemployment problem} since it leaves some mobiles users unassociated (unmatched).
There is the need for a control of the players incentives for some equilibrium points with satisfying properties, in terms of unemployment in the present case.
In other words, since the players have the incentive to match in a one-to-one form, one needs to control the underlying cooperative game so as to provide new incentives for a suitable many-to-one form as an equilibrium.

\section{Mechanism for controlled matching game} \label{sec:mechanism}

In order to tackle the {\it unemployment problem}, we propose in this section a mechanism to control the players incentives for coalitions (see Figure \ref{fig:generalblockdiagram}). This mechanism is made of three steps. We start by considering for every AP the set of acceptable mobile users, i.e., the mobile users with non zero data rate with this AP. 
In the first step (block \textbf{LB}), APs share the load defined in number of users.
This results in objective quotas that should be enforced by the mechanism. 
The second step (blocks $\mathbf{\Omega}$ and $\mathbf{\Phi}$) is a controlled coalition game designed so as to provide the players the incentives to form coalitions with cardinalities given by the quotas and reducing heterogeneity (and thus reducing the anomaly in the IEEE 802.11). The third step (block $\mathbf{\mu}$) is a decentralized coalition formation (or matching) algorithm which results in a stable structure induced by the individual preferences influenced  by the controlled coalitional game. 
\begin{figure*}[t]
	\centering

	\tikzstyle{int}=[draw, minimum size=0.8cm]
	\tikzstyle{init} = [pin edge={-to,thin,black}]
	\tikzstyle{test} = [pin edge={to-,thin,black}]
	
	\begin{tikzpicture}[node distance=3cm,auto,>=latex']
	    
	    \node [int,pin={[test]above:$\mathcal{F}$}] (z) {LB};
	    \node (b) [left of=z,node distance= 1.5cm, coordinate] {z};
	    \node [int,pin={[test]above:$v$}] (a)[right of=z] {$\Omega$};
	    \node [int,pin={[test]above:$(u_{i}(.))_{i\in\mathcal{N}}$},pin={[init]below:$(s_{i,C}^{*})_{i\in C,C\in\mathcal{C}}$}] (c) [right of=a] {$\Phi$};
	    \node [int] (d) [right of=c] {$\mu$};
	    \node [int] (e) [right of=d] {MAC};
	    \node [coordinate] (end) [right of=e]{};
	    
	    \path[->] (b) edge node {$\mathcal{W}$} (z);
	    \path[->] (z) edge node {$\mathbf{\hat{q}}$} (a);
	    \path[->] (a) edge node {$\tilde{v}$} (c);
	    \draw[->] (c) edge node {$(u_{i,C}^{*})_{i\in C,C\in\mathcal{C}}$} (d) ;
	    \draw[->] (d) edge node {$\mu$} (e) ;
	    \draw[->] (e) edge node {$(r_{i,C})_{i\in C,C\in\mu}$} (end) ;

	\draw [color=gray,thick](-2,-1.60) rectangle (10,1.60);
	\node at (-2,1.80) [above=10mm, right=0mm] {MECHANISM};
	\draw [color=gray,thick](10.8,-1.60) rectangle (15,1.60);
	\node at (10.8,1.80) [above=10mm, right=0mm] {802.11 MAC};
	\end{tikzpicture}
	\caption{Block diagram of the mechanism in the most general form. The APs share the load in the block \textbf{LB} which gives the APs' objectives $\mathbf{\hat{q}}$. The characteristic function $v$ of the original coalition game is controlled in $\mathbf{\Omega}$ and gives the modified characteristic function $\tilde{v}$. The Nash bargaining $\mathbf{\Phi}$ is played in each coalition for the allocation of the worth of the coalition among its members. The players then emit their preferences over the coalitions on the basis of their shares and enter a stable matching mechanism in block $\mathbf{\mu}$. This block outputs an AP-user association $\mu$. Finally, in the block $\mathbf{MAC}$ the nodes transmit their packets according to the unmodified IEEE 802.11 MAC protocol.}
	\label{fig:generalblockdiagram}
\end{figure*}
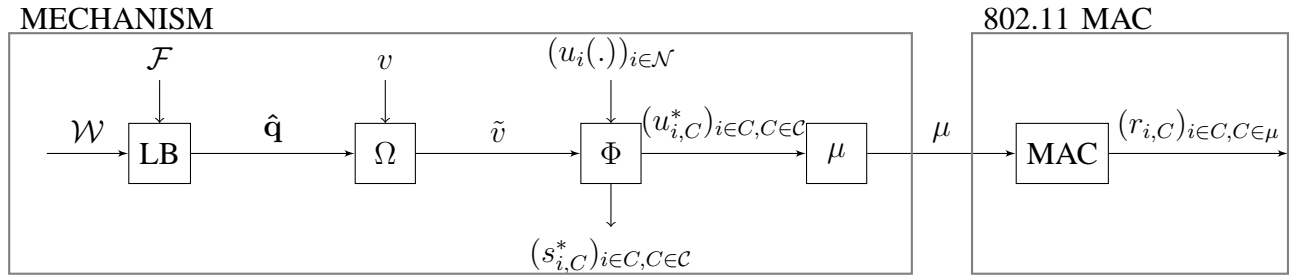

Our mechanism can be implemented as a virtual connectivity management layer on top of the IEEE 802.11 MAC protocol. Mobile users and APs form coalitions based on the "virtual rates" provided by this virtual layer. Once associated, users access the channel using the unmodified 802.11 MAC protocol. 



\subsection{ Load Balancing}\label{subsec:loadbalancing}

The first step of the mechanism is a load balancing.
This step outputs a quota vector of the form $\mathbf{\hat{q}} = (\hat{q}_{1},\ldots,\hat{q}_{F})$ that defines the size of the coalitions the players should be incentivized to form with each AP.
In other words, $\mathbf{\hat{q}}$ gives the number of connections the players should be incentivized to create with each AP.
As in \cite{Tou2015b}, in numerical implementation (see Section~\ref{sec:wlan}), we take a Nash bargaining based decentralized load balancing scheme between the APs to share the users covered by several APs.
This scheme achieves the proportional fair allocation in the utility space.
Nevertheless, any load balancing scheme can be used in this mechanism. 

\subsection{Controlling coalition game}\label{subset:mechanism}

The second step of the mechanism is the control of the coalition game.
The control step of the mechanism tackles the problem of the control of the set of stable matchings.  
We observed that when a coalition game is defined by a characteristic function and a sharing rule inducing sub-additive strictly positive individual payoffs (except for coalitions of size one or those containing players with zero data rates), the stable structures to be formed are made of coalitions of size two.
This step of the mechanism develops an analytical framework and methodology for the control of the equilibria by the way of a control over the players' incitations for individual strategies.
	\begin{definition}[Controller]
				The controller is any entity (player or other) having the legitimacy and ability to change the definition of the game (players, payoffs, worths, information, coalitions). 
	\end{definition}
	The \textit{controller} may not be taking part in the game (e.g. the network operator in a wireless network, the government for a firms and workers association problem) or any player of the game with some kind of additional decisional abilities.
In other words, it may be any entity having the ability to create or modify the individual incitations of the players for some strategy and thus the ability to change the definition of the game.
These changes in the definition of the game in view of manipulating the players' equilibria strategies are called control transformation,
	\begin{definition}[Control transformation]
				A \textbf{control transformation} $\Omega$ is a mapping from the set of coalition games in characteristic form in itself.
	\end{definition}
In the purpose of this paper it is sufficient to restrict the definition of the control transformations to the domain of coalition games in characteristic form.	
In fact, we further assume that the controller cannot arbitrarily move from one game to another without constraints.
We assume that he or she can influence the equilibria by partial changes in the definition of the game (characteristic function, individual payoffs, ...) but can neither change the fundamental rules of the game (e.g. the rules of matching games) nor some essential elements such as the players taking part in the game or their strategy spaces.
If $\Gamma$ is a coalition game in characteristic form, then $\Omega(\Gamma)$ is a coalition game in characteristic form modified by the controller according to its (constrained) abilities.
The limits of the abilities of such a controller are to be chosen by the game theorist or the designer of the mechanism so as to satisfy the fundamental hypothesis and description of the system he is looking at.
The controller and the control transformation may be defined as the result of another game at a higher level (see the application with bargaining APs for quotas).
As an example of work on the design of an incitations operator, Auman and Kurz \cite{Aum1977} assess the problem of designing the joint taxation and redistribution scheme in the framework of a political majority-minority game. The majority is the controller and the incitations are induced by a multiplicative tax over the worths of the coalitions.
	

In Appendix~\ref{app:examples-control}, we give two simple example of the mechanism we propose to control the player's individual incentives.

We now search for operators modifying the characteristic function $v$ so as to provide  players the incentives to form stable structures with coalitions of sizes $\mathbf{\hat{q}}$.

An important lever for controlling our matching game and designing operator $\mathbf{\Omega}$ is the fear-of-ruin (FoR). Formally, the FoR of user $i$ in coalition $C$ is defined as: 
\begin{equation}
\chi_{i}(s_{i,C})\triangleq\frac{u_{i}(s_{i,C})}{u_{i}^{'}(s_{i,C})}. 
\end{equation}
The FoR of coalition $C$ is obtained as the inverse of the Lagrange multiplier associated to the constraint $\sum_{i\in C}s_{i,C}\leq v(C)$ at the optimum of the Nash bargaining optimization problem~(\ref{eq:NBresourcealloc}). Two interesting characteristics of the FoR are that (i) in a coalitional game with Nash bargaining as sharing rule, the FoR is constant over the players in a coalition, i.e., $\chi_{i}(s_{i,C}) = \chi_{C}$ $\forall i\in C$ at the bargaining solution point $s_{i,C}$ and (ii) with concave increasing utility functions, the individual payoffs increase in the common FoR~\cite{Pyc2012}. Thus, the players have the incentives to form coalitions maximizing their FoR. In terms of control opportunities, this introduces the FoR as a lever to control the set of individual payoff-based incentives for coalitions. As an example, assume two coalitions $C$ and $C'$ and their FoRs: $\chi_{C}<\chi_{C'}$. Players in $C\cap C'$ prefer $C'$ to $C$. Changing the values of the FoRs to obtain $\chi_{C}>\chi_{C'}$ changes the individual incentives of these players so that they now prefer $C$ to $C'$. 

\begin{restatable}{proposition}{propositioncontrol}\label{prop:controlincentives}
Assume a coalition game $\Gamma = (\mathcal{F}\cup\mathcal{W}, v, \{u_{i}\}_{i\in N})$ in characteristic form with the Nash bargaining sharing rule over $v(C)$
for every coalition $C$ in $\mathcal{C}$. Furthermore assume strictly increasing and concave utility functions\footnote{Such utility functions are bijective and thus injective. Theorem \ref{the:solutionnashproduitutilities} applies.} $u_{i}:\mathbb{R}^{.+}\rightarrow\mathbb{R}^{+}, i\in\mathcal{N}$. 
The set of transformations $\Omega$ from the set of characteristic functions in itself that provide the players the incentive for some subset $\mathcal{C}'$ of coalitions in $\mathcal{C}$ must satisfy:
	\begin{equation}
		F_{C'}\circ\Omega(v)(C') < F_{C}\circ\Omega(v)(C) \quad \forall C'\in\mathcal{C}',\forall C\in\mathcal{C}\backslash \mathcal{C}'
	\end{equation}
	s.t. $C'\cap C \neq \emptyset$ and where $F_{C} = \left(\sum\limits_{i\in C}\left(\frac{u_{i}^{'}}{u_{i}}\right)^{-1}\right)^{-1}$.
\end{restatable}
\begin{proof}
	See Appendix~\ref{appendix:proofs}.
\end{proof}
In order to derive our last result, we need to define the concept of single-peaked preferences.
Let $X = \{x_{1},\ldots,x_{n}\}$ denote a finite set of alternatives, with $n\geq 3$.
\begin{definition}[Peak of preferences, \cite{Esc2008}]
	A preference relation $\succ$ on $X$ is a linear order on $X$. The peak of a preference relation $\succ$ is the alternative $x^{*} = peak(\succ)$ such that $x^{*}\succ x$ for all $x\in X\backslash\{x^{*}\}$. 

\end{definition}
\begin{definition}[Single-Peaked preferences, \cite{Esc2008}]
	An axis $O$ (noted by $>$) is a linear order on $X$.
	Given two alternatives $x_{i},x_{j}\in X$, a preference relation $\succ$ on $X$ whose peak is $x^{*}$, and an axis $O$, we say that $x_{i}$ and $x_{j}$ are on the same side of the peak of $\succ$ iff one of the following two condition is satisfied: (i) $x_{i}> x^{*}$ and $x_{j}>x^{*}$; (ii) $x^{*}>x_{i}$ and $x^{*} > x_{j}$.\\
	A preference relation $\succ$ is single-peaked  with respect to an axis $O$ if and only if for all $x_{i}$, $x_{j}\in X$ such that $x_{i}$ and $x_{j}$ are on the same side of the peak $x^{*}$ of $\succ$, one has $x_{i}\succ x_{j}$ if and only if $x_{i}$ is closer to the peak than $x_{j}$, that is, if $x^{*}>x_{i}>x_{j}$ or $x_{j}>x_{i}>x^{*}$. 
\end{definition}
We use the discrete version of this definition over $\mathbb{N}^{+}$.
We immediately obtain the following corollary,
\begin{restatable}{corollary}{corollarycontrol}\label{cor:controlincentives}
Assume a coalition game $\Gamma = (\mathcal{F}\cup\mathcal{W}, v, \{u_{i}\}_{i\in N})$ in characteristic form with the Nash bargaining sharing rule over the $v(C)$
in every coalition $C\in \mathcal{C}$. Furthermore assume strictly increasing and concave utility functions $u_{i}:\mathbb{R}^{+}\rightarrow\mathbb{R}^{+}, i\in\mathcal{N}$. 
The set of transformations $\Omega$ from the set of characteristic functions in itself that induce single-peaked preferences (peak at $\hat{q}_{f}$) in cardinalities over the coalitions with an AP $f\in\mathcal{F}$ must satisfy:
	\begin{equation}\label{eq:fearofruinrequirement}
			\mathop{\max}\limits_{\substack{C\in\mathcal{C}_{f}\\ s.t. |C|=q}}F_{C}\circ\Omega(v)(C)
			<
			\mathop{\min}\limits_{\substack{C\in \mathcal{C}_{f}\\ s.t. |C|= q+1}}F_{C}\circ\Omega(v)(C),\quad \forall {q}\geq \hat{q}_{f}
	\end{equation}
	and
	\begin{equation}
			\mathop{\max}\limits_{\substack{C\in\mathcal{C}_{f}\\ s.t. |C|=q}}F_{C}\circ\Omega(v)(C)
			<
			\mathop{\min}\limits_{\substack{C\in \mathcal{C}_{f}\\ s.t. |C|= q-1}}F_{C}\circ\Omega(v)(C),\quad \forall {q}\leq \hat{q}_{f}
	\end{equation}
	where $F_{C} = \left(\sum\limits_{i\in C}\left(\frac{u_{i}^{'}}{u_{i}}\right)^{-1}\right)^{-1}$.
\end{restatable}
\begin{proof}
	See Appendix~\ref{appendix:proofs}.
\end{proof}

\subsection{Access Point Association }

The third step of the mechanism is the joint resource allocation and users association (matching) game where the players (APs and mobile users) share the resource in the coalitions according to a Nash bargaining and then match with each others. 
The coalition game played has been described in Section~\ref{sec:matching} and Section~\ref{sec:existencestable}.
This step corresponds to the blocks $\Phi$ and $\mu$ of the block diagram in Figure~\ref{fig:generalblockdiagram}.

\subsubsection{Stable Matching Mechanism}\label{sec:BDAA}

\begin{algorithm}[h]
{\footnotesize
\KwData{For each AP: The set of acceptable (covered) users and AP-user data rates.\\ For each user: The set of acceptable (covering) APs.}
\KwResult{A core stable structure $\mathcal{S}$}
\Begin{
		\textit{Step 1: Initialization}\;
			\Indp
			\textbf{Step 1.a:} All APs and users are marked {\it unengaged}.  	$L(f)=L^*(f)=\emptyset$, $\forall f$\;
			\textbf{Step 1.b:} Every AP $f$ computes possible coalitions with its acceptable users, the respective users payoffs and emits its preference list $P^{\#}(f)$\;
			\textbf{Step 1.c:} Every AP $f$ transmits to its acceptable users the highest payoff they can achieve in coalitions involving $f$\;
			\textbf{Step 1.d:} Every user $w$ emits its reduced list of preference $P'(w)$\;				
				
		\Indm 	 
		\textit{Step 2 (BDAA)}\;
			\Indp
				\textbf{Step 2.a, Mobiles proposals:} According to $P'(w)$, every unengaged user $w$ proposes to its most preferred acceptable AP for which it has not yet proposed. If this AP was engaged in a coalition, all players of this coalition are marked {\it unengaged} \;
				\textbf{Step 2.b, Lists update:} Every AP $f$ updates its list with the set of its proposers: $L(f) \longleftarrow L(f)\cup\{\text{proposers}\}$ and $L^{*}(f) \longleftarrow L(f)$\;
				\textbf{Step 2.c, Counter-proposals:} Every AP $f$ computes the set of coalitions with users in the dynamic list $L^{*}(f)$ and counter-proposes to the users of their most preferred coalition according to $P^{\#}(f)$\;
				\textbf{Step 2.d, Acceptance/Rejections:} Based on these counter-proposals and the best achievable payoffs offered by APs in Step 1.c to which they have not yet proposed, users accept or reject the counter-proposals\;
				\Indp
				\textbf{Step 2.e:} If all users of the most preferred coalition accept the counter-proposal of an AP $f$, all these users and $f$ defect from their previous coalitions\; all players of these coalitions are marked {\it unengaged}\; users that have accepted the counter-proposal and $f$ are marked {\it engaged in this new coalition}\;
				\textbf{Step 2.f:} Every unengaged AP $f$ updates its dynamic list by removing users both having rejected the counter-proposal and being engaged to another AP: \\$L^{*}(f) \longleftarrow L^{*}(f)\backslash\{\text{engaged rejecters}\}$\;
				\Indm 
				\textbf{Step 2.g:} Go to Step 2.c while the dynamic list $L^{*}$ of at least one AP has been strictly decreased (in the sense of inclusion) in Step 2.f\;
				\textbf{Step 2.h:} Go to Step 2.a while there are unengaged users that can propose\;		
				\textbf{Step 2.i:} All players engaged in some coalition are matched.	
}
\caption{Backward Deferred Acceptance}
\label{algo:BDAA}
}
\end{algorithm}

We now show that a modified version of the Gale and Shapley's deferred acceptance algorithm in its college-admission form with APs preferences over groups of users and users preferences over individual APs is a stable matching mechanism for the many-to-one matching games with complementarities, peer effects considered in this paper  (see Algorithm 1: Backward Deferred Acceptance). 

BDAA is similar to the DAA in many aspects. It involves two sets of players that have to be matched. Every player from one side has a set of unacceptable players from the other side. In our case, an AP and a mobile user are acceptable to each others if the user is under the AP coverage. As in DAA, the algorithm proceeds by proposals and corresponding acceptances or rejections. The main difference resides in the notion of counter-proposals, introduced to tackle the problem of complementarities. 

\begin{figure*}[t]
	\centering
	\tikzstyle{block} = [draw, rectangle, 
	    minimum height=1cm, minimum width=2cm]
	\tikzstyle{sum} = [draw, circle, node distance=3cm]
	\tikzstyle{input} = [coordinate]
	\tikzstyle{output} = [coordinate]
	\tikzstyle{pinstyle} = [pin edge={to-,thin,black}]
	
	\begin{tikzpicture}[auto, node distance=2cm,>=latex']
	
	    \node [input, name=input] {};
	    \node [block, right of=input] (preferencesfirms) {$\mathbf{P^{\#}}$};
	    \node [block, right of=preferencesfirms, node distance=3cm] (preferencesworkers) {$\mathbf{P^{'}}$}; 
	    \node [sum, node distance=1.5cm, right of=preferencesworkers] (sum) {};
	    \node [block, right of=sum, node distance=1.5cm] (proposals) {Proposals};
	    \node [sum, right of=proposals,node distance=1.5cm] (backcounterproposals) {};   
	    \node [block, right of=backcounterproposals,node distance=2.5cm] (counterproposals) {Counter-proposals};
	    	\coordinate [below of=counterproposals, node distance=1.2cm] (tmp);
		\coordinate [below of=tmp, node distance=0.7cm] (tmpbis);
		
	    \node [sum, node distance=2.5cm, right of=counterproposals] (sumend) {};
	    \node [sum, node distance=1cm, right of=sumend] (sumendbis) {};
	             
	    \node [output, right of=sumend] (output) {};
	
	    \draw [draw,->] (input) -- (preferencesfirms);
	    \draw [draw,->] (preferencesfirms) -- (preferencesworkers);
	    \draw [->] (preferencesworkers) -- (sum);
	    \draw [->] (sum) --  (proposals);
	    \draw [->] (proposals) --  (backcounterproposals);
	    \draw [->] (backcounterproposals) -- (counterproposals);
	    \draw [->] (counterproposals) --  (sumend);
	    \draw [->] (sumend) -- (sumendbis);
	    \draw [->] (sumendbis) -- node [name=mu] {$\mu$}(output);
	    \draw [->] (sumend) |- node[pos = 0.9, above] {\footnotesize counter-proposing loop}  (tmp) -| (backcounterproposals);
	    
	    \draw [->] (sumendbis) |- node[pos = 0.9, above] {\footnotesize proposing loop} (tmpbis) -|  (sum);
	\end{tikzpicture}
	
	\caption{Block diagram of the BDAA.}
	\label{fig:blockdiagramBDAA}
\end{figure*}
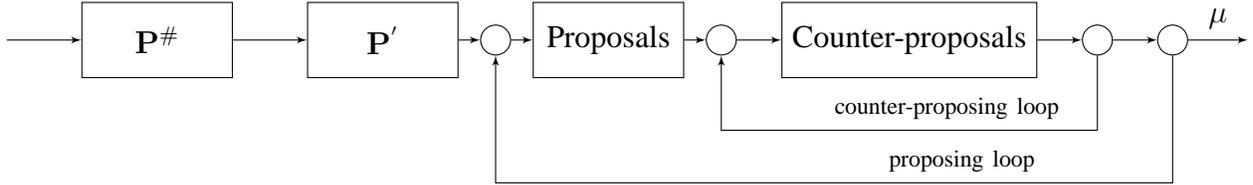

The block diagram representation of the algorithm is shown in Figure (\ref{fig:blockdiagramBDAA}).
In block $\mathbf{P^{\#}}$ the APs emit their preferences over the coalitions. In block $\mathbf{P^{'}}$ the mobiles emit their preferences over the APs. In block \textit{Proposals} the mobiles propose to the APs. In block \textit{counter-proposals} the APs counter-propose. 
The counter-proposing round continues up to convergence. The next proposing round starts.\\
We enter the details of the algorithm.
Having the information of the data rates with users under their coverage, APs are able to compute all the possible coalitions they can form and the corresponding allocation vectors (throughputs). They can thus build their preference lists (Steps 1b). Then, every AP $f$ transmits to each of its acceptable users  the maximum achievable throughput (based on MAC layer and virtual mechanism) it can achieve in the coalitions it can form with $f$ (Step 1.c). Every user $w$ can thus build its reduced list of preferences over individual APs: $w$ prefers $f_{i}$ to $f_{j}$ if the maximum achievable throughput with $f_{i}$ is strictly greater than its maximum achievable throughput with $f_{j}$ (Step 1.d). BDAA then proceeds by rounds during which users make proposals, AP make counter-proposals and users accept or reject (from Step 2.a to Step 2.h). Every AP that receives a new proposal shall reconsider the set of its opportunities and is thus marked unengaged (Step 2.a). $L(f)$ is the list of all users that have proposed at least once to AP $f$. $L^*(f)$ is a dynamic list that is reinitialized to $L(f)$ before every AP counter-proposal (Step 2.b). In each round of the algorithm, every unengaged user proposes to its most preferred AP for which it has not yet proposed (Step 2.a). Every AP receiving proposals adds the proposing players to its cumulated list of proposers and reinitializes its dynamic list (Step 2.b). Using $P^{\#}(f)$ it then searches for its most preferred coalition involving only users from the dynamic list and emits a counter-proposal to these users.
This counter-proposal contains the throughput every user can achieve in this coalition (Step 2.c). Each user compares the counter-proposals it just received with the best achievable payoffs obtained with the APs it has not proposed to yet (Step 2.d).
If one of these best achievable payoffs is strictly greater than the best counter-proposal, the users rejects the counter-proposals and continues proposing (Step 2.d, Step 2.h). Otherwise, the user accepts its most preferred counter-proposal (Step 2.d).
Given a counter-proposal, if all users accept it, then they are engaged to the AP. All coalitions in which these users and the AP were engaged are broken and their players are marked unengaged (Step 2.e). If at least one user does not agree, then the AP is unengaged (Step 2.e),  it updates its dynamic list by removing the mobiles having rejected its counter-proposal and being engaged to another AP (Step 2.f). The counter-proposals continues up to the point when no AP can emit any new counter-proposal (Step 2.g). The current round ends and the algorithm enters a new round (Step 2.h). The algorithm stops when no more users are rejected (Step 2.h). A stable matching is obtained (Step 2.i).

\begin{restatable}{proposition}{convergenceBDAA}\label{prop:ConvBDAA}
Given a many-to-one matching game, BDAA converges, i.e., outputs a matching in a finite number of steps.
\end{restatable}
\begin{proof}
	See Appendix~\ref{appendix:proofs}.
\end{proof}
\begin{restatable}{proposition}{stabilityBDAA}\label{prop:StabBDAA}
	Suppose the family of coalitions $\mathcal{C}$ as defined in \eqref{eq:coalitionC}, and a sharing rule as defined in proposition \ref{prop:stability}. Furthermore assume a tie-breaking rule such that there is no indifference (strict preferences). 
	 BDAA converges to the unique core stable matching.
\end{restatable}
\begin{proof}
	See Appendix~\ref{appendix:proofs}.
\end{proof}
\begin{restatable}{proposition}{complexityBDAA}\label{prop:CompBDAA}
The complexity of BDAA is $O(n^5)$ in the number of proposals of the players, where $n=\max(F, W)$.
\end{restatable}
\begin{proof}
	See Appendix~\ref{appendix:proofs}.
\end{proof}

In Appendix~\ref{app:BDAA-interpretations}, we provide an interpretation of BDAA in the economic framework. In Appendix~\ref{appendix:Example}, we give an example of application of the BDAA.

\section{Numerical Results} \label{sec:wlan}

\subsection{Simulations Parameters and Scenarios}

The numerical computations are performed under the assumption of equal packet sizes and saturated queues (each node has always packets to transmit). Under this assumption the sharing rule is equal sharing. Analytical expressions of the throughputs (individual and total throughputs) are taken from~\cite{Kum2007} with the parameters of Table \ref{table:simulationparameters}. 
\begin{figure}[b]
\begin{center}
  \subfigure[Scenario 1]{\epsfig{figure=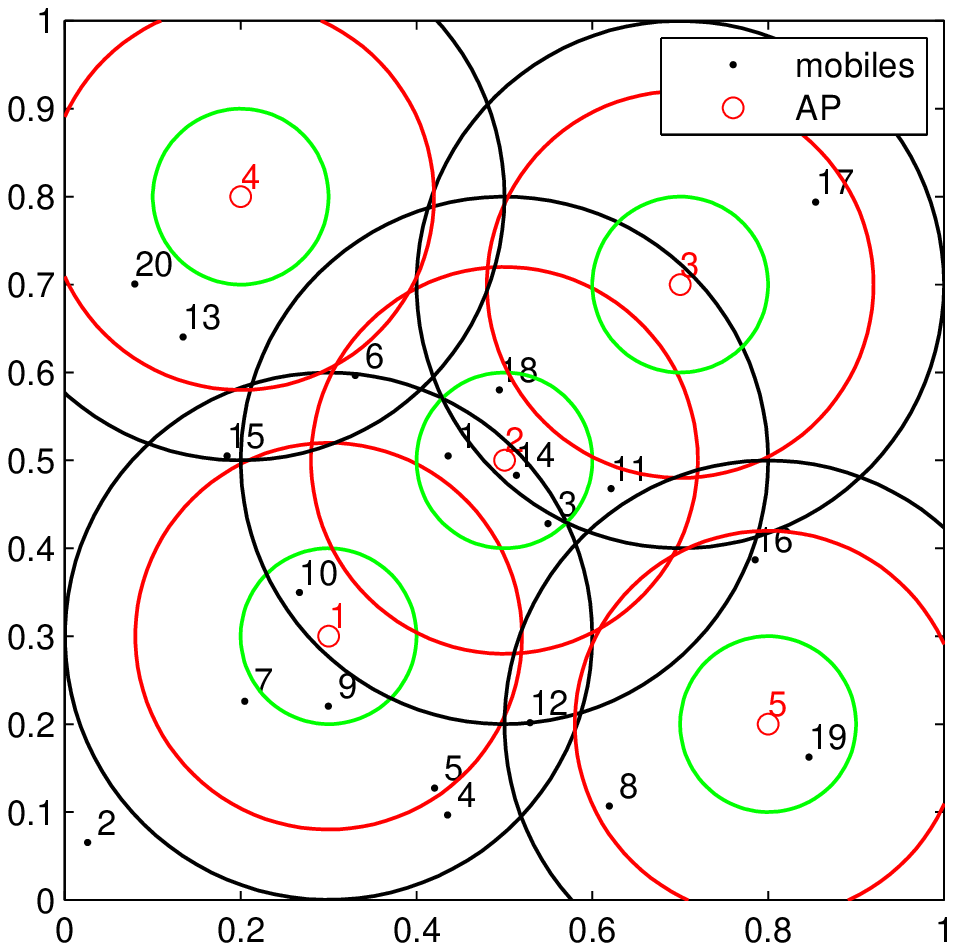,width=7cm}}
  \subfigure[Scenario 2]{\epsfig{figure=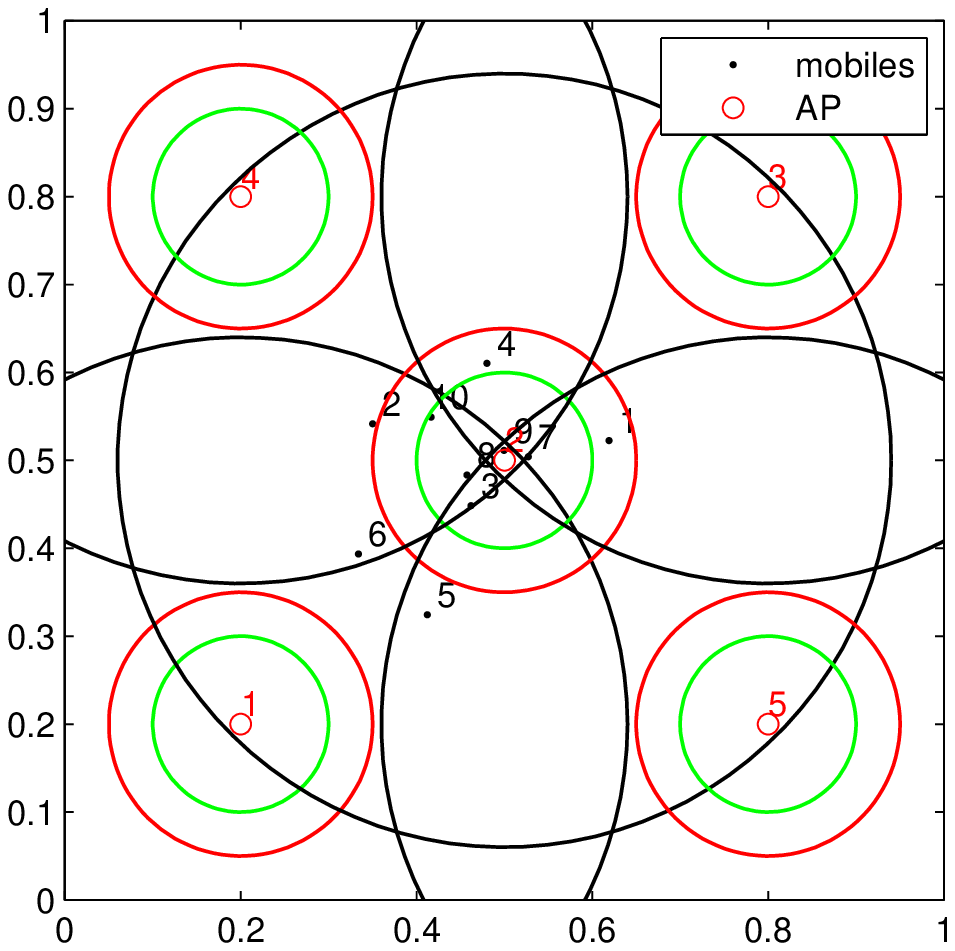,width=7cm}}
\end{center}
 \caption{Scenario 1 (left): A spatial distribution of  APs (smallest red circles) $\mathcal{F} = \{f_{1},\ldots,f_{5}\}$ and devices (black points) $\mathcal{W} = \{w_{1},\ldots,w_{20}\}$. Scenario 2 (right): A spatial distribution of  APs (smallest red circles) $\mathcal{F} = \{f_{1},\ldots,f_{5}\}$ and devices (black points) $\mathcal{W} = \{w_{1},\ldots,w_{10}\}$. Circles show the coverage areas corresponding to different data rates.} 
 \label{fig:scenarios}
\end{figure}
We further assume that a node compliant with a IEEE 802.11 standard (in chronological order: b, g, n) is compliant with earliest ones. By convention, if all nodes of a cell have the same data rate, we use the MAC parameters of the standard whose maximum physical data rate is the common data rate. Otherwise, we use the MAC parameters of the standard whose maximum physical data rate is the lowest data rate in the cell.

Assume the spatial distributions of nodes of Figure \ref{fig:scenarios}. The first scenario (a) shows the case of $5$ APs with a uniform spatial distribution of 20 mobile users. The second scenario (right) has non-uniform distribution of 10 mobile users in the plane. The green (inner), red (intermediate) and black (outer) circles show the spatial region where the mobiles achieve a data rate of 300~Mbits/s, 54~Mbits/s and 11~Mbits/s respectively. Scenario 2 exhibits a high overlap between AP coverages. 
\begin{table}[b]
\vspace{0.3cm}
\footnotesize
\begin{center}
\begin{tabular}{|c|c|c|c|c|c|c|}
\cline{2-4}
\multicolumn{1}{c|}{} &802.11n& 802.11g& 802.11b&\multicolumn{1}{c}{} \\
\hline
Parameter&\multicolumn{3}{c|}{value} &unit\\
\hline
$\Theta$& \{300, 54, 11\}& \{54, 11\}& \{11\}& Mbits/s\\
\hline
slot duration&9&9&20&$\mu$s\\
\hline
$T_0$&3&5&50&slots\\
\hline
$T_C$&2&10&20&slots\\
\hline
$L$&8192&8192&8192&bits\\
\hline
$K$&2&2&2&\multicolumn{1}{|>{\columncolor{kugray5}}c|}{}\\
\cline{1-4}
$b_0$&16&16&16&\multicolumn{1}{|>{\columncolor{kugray5}}c|}{}\\
\cline{1-4}
$p$&2&2&2&\multicolumn{1}{|>{\columncolor{kugray5}}c|}{}\\
\hline
\end{tabular}
\end{center}
\caption{Simulation Parameters.}
\label{table:simulationparameters}
\end{table}%

\subsection{Numerical Work}

\subsubsection{No mechanism}
We show in Figure~\ref{fig:nocost} a stable matching.
No associated player has an incentive to deviate and form a coalition of size superior to two. The figure shows the natural incentives of the system in forming low cardinalities coalitions with good compositions. 
\begin{figure}[h]
	\centering
	\includegraphics[width = 7cm]{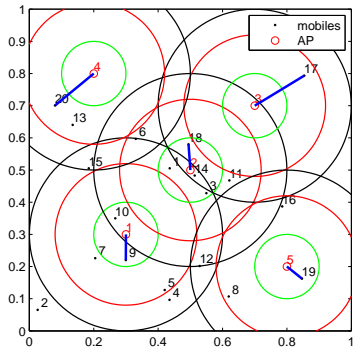}
	 \caption{A stable matching in the uncontrolled case.}
	 \label{fig:nocost}
\end{figure}
This can also be observed on Figure~\ref{fig:throughputs} which shows the individual throughputs obtained in the coalitions.
The coalitions are sorted by cardinalities from low to high.
In plot~(a) no mechanism is used. In plot~(b) a gaussian tax rate is applied. See Section~\ref{subsec:GaussianTaxRate}.

\begin{figure}[t]
\begin{center}
  \subfigure[Individual throuputs vs. coalition index.]{\epsfig{figure=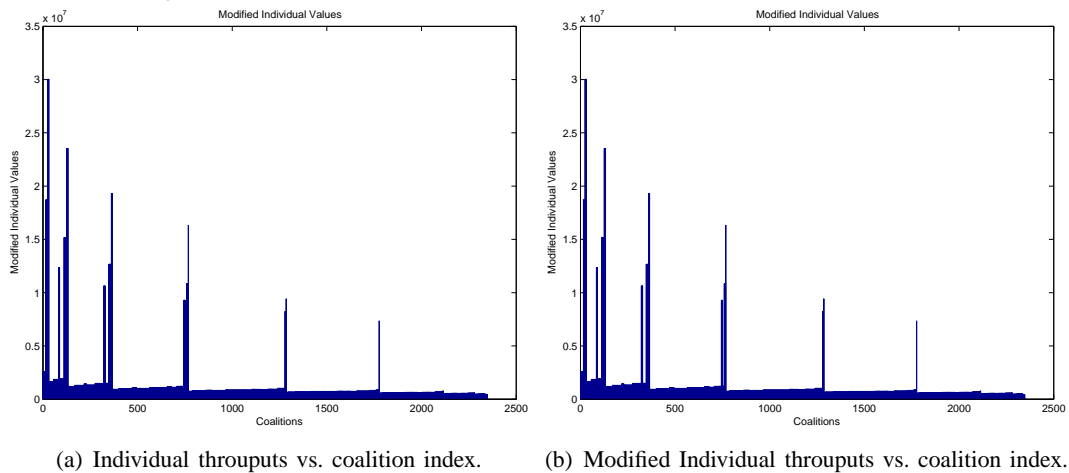,width=7cm}}
  \subfigure[Modified Individual throuputs vs. coalition index.]{\epsfig{figure=modified_individual_values.eps,width=7cm}}\\
\end{center}
 \caption{(a)Scenario 1. Structure of the payoffs in the uncontrolled matching game. (b)Scenario 1. Structure of the payoffs in the controlled matching game with a multiplicative tax rate of variance $\sigma_{f} = 0.3, \forall f\in\mathcal{F}$.}
 \label{fig:throughputs}
\end{figure}
 Figure~\ref{fig:nocost} and Figure~\ref{fig:throughputs}~(a)  show the natural incentives of the system in forming low cardinalities coalitions with good compositions. As a result, a one-to-one matching is obtained. Using our mechanism, this structure of throughputs will be changed (as in Figure \ref{fig:throughputs}~(b)) to move the incentives according to $\mathbf{\hat{q}}$ and thus provide the players the incentives to associate according to a many-to-one matching rather than a one-to-one.

\subsubsection{Gaussian Tax Rate in Cardinalities}\label{subsec:GaussianTaxRate}
As an example of family of cost functions, we can use multiplicative symmetric unimodal cost functions. 
The multiplicative cost functions are commonly called tax rates and are defined such that
for any AP $f\in\mathcal{F}$ and any coalition $C$ containing $f$, we must have:
\begin{equation}
	\tilde{v}(C) = \mathbf{\Omega}(v(C))\triangleq c_f(|C|)v(C)
\end{equation}
We particularly consider Gaussian tax rates such that:
\begin{equation}
	\tilde{v}(C) =   e^{-\frac{(|C|-\hat{q}_{f})^{2}}{2\sigma_{f}^{2}}}v(C)
\end{equation}
where $\sigma_{f}$ is the variance of the function $c_f$.
The Gaussian cost function is convenient in the sense that it does not penalize the mean-sized coalitions and it provides a great amount of flexibility by the way of its variance. Decreasing or increasing the variance $\sigma_{f}$ indeed allows for a strict or relaxed control of the incentives for the objective quotas.
\begin{figure}[t]
\begin{center}
  \subfigure[Stable matching resulting from Gaussian costs and BDAA.]{\epsfig{figure=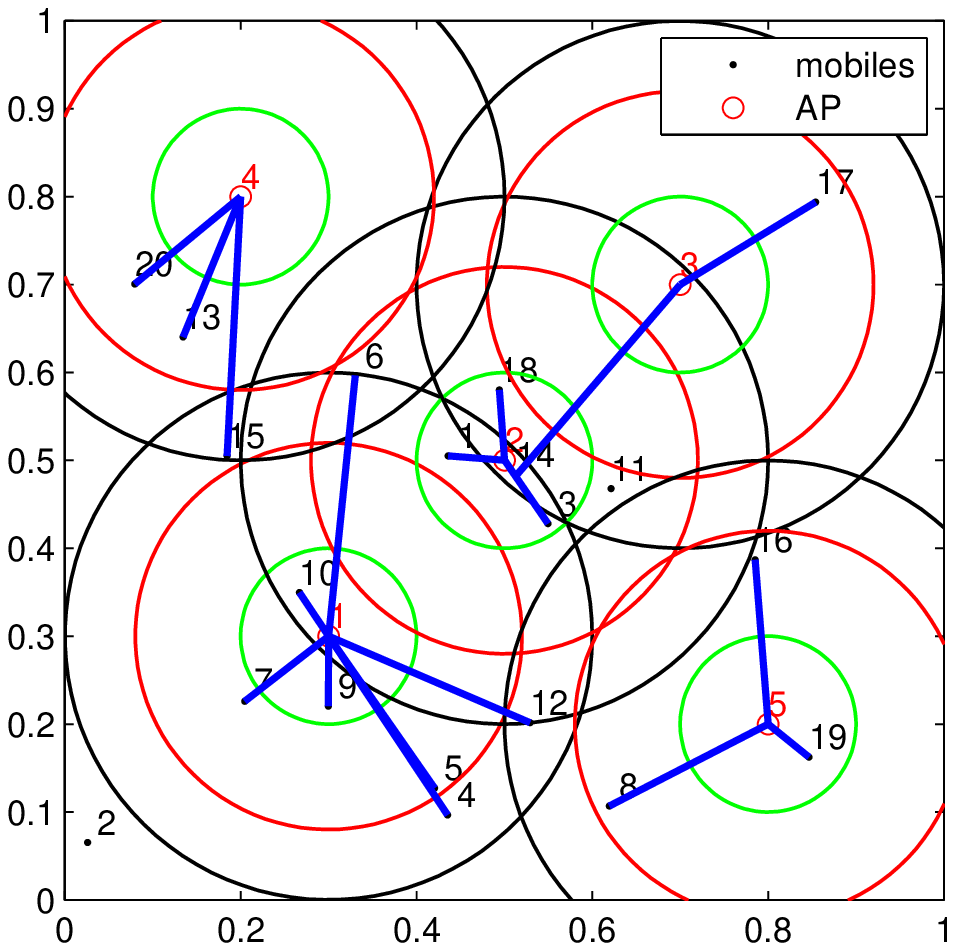,width=5cm}}
  \subfigure[A global optimum association with Gaussian costs.]{\epsfig{figure=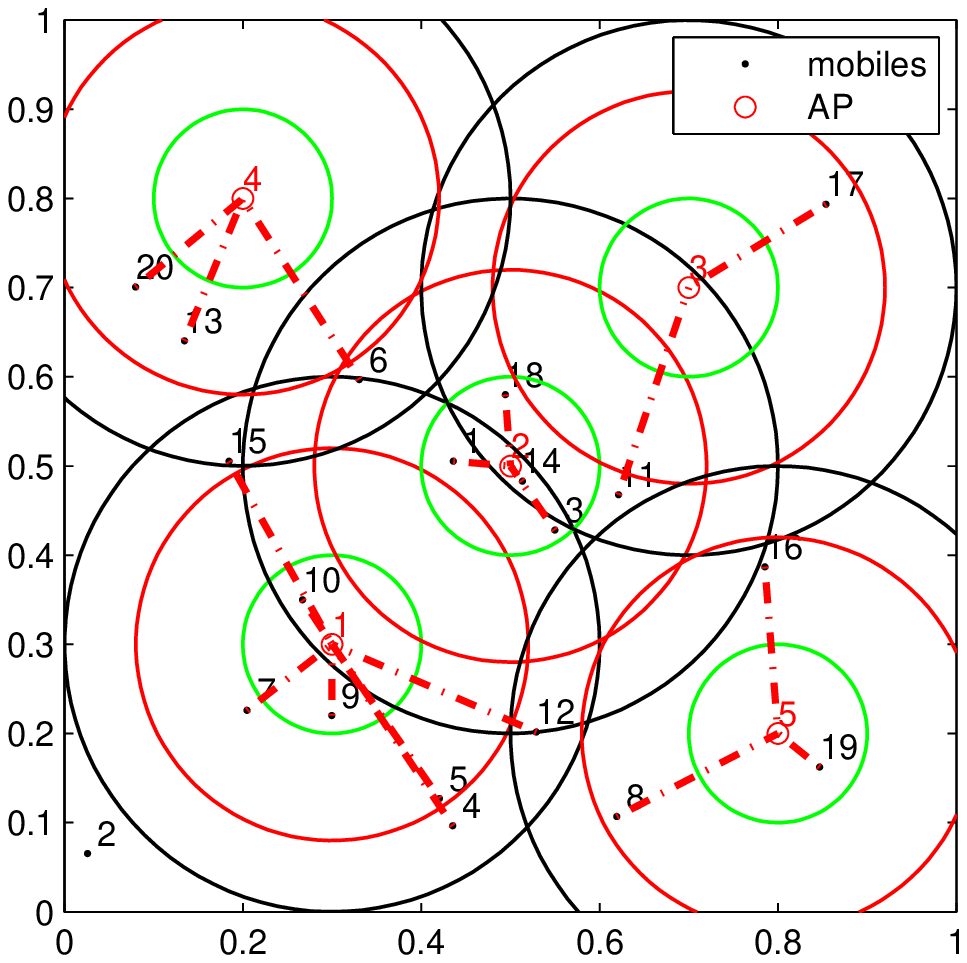,width=5cm}}
  \subfigure[A global optimum association without costs.]{\epsfig{figure=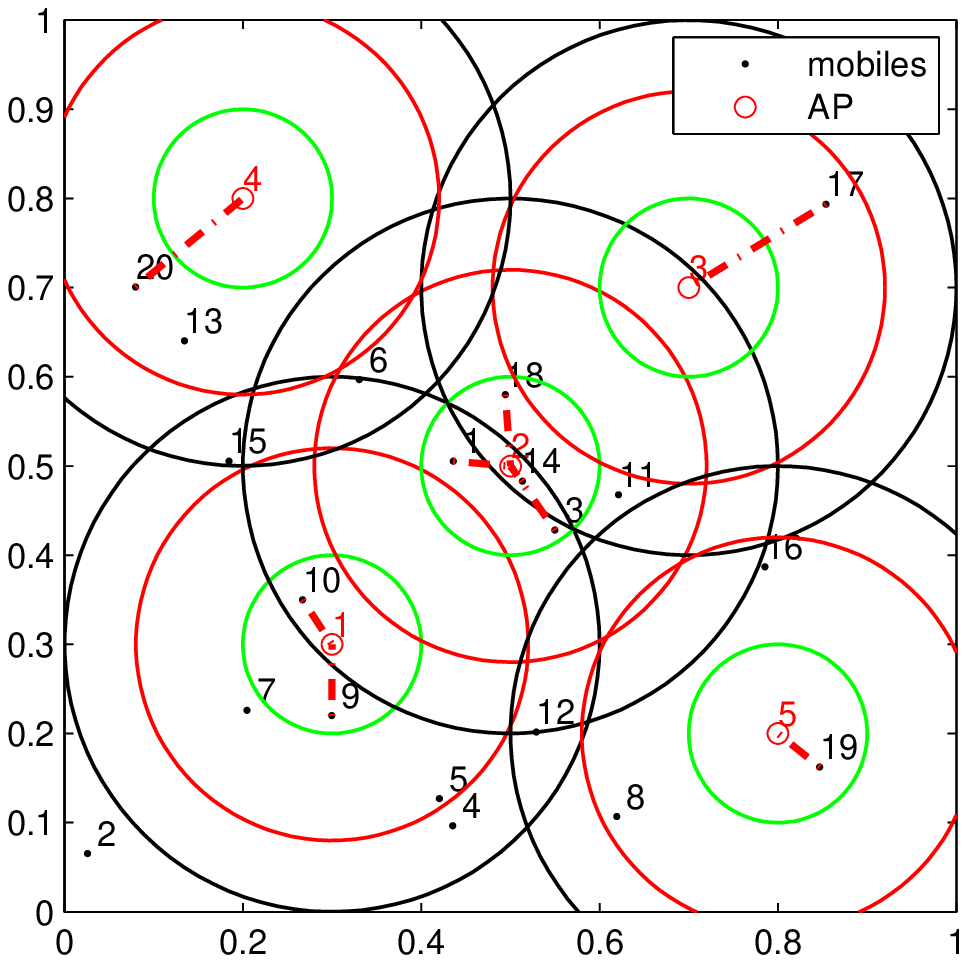,width=5cm}}
\end{center}
 \caption{Controlled matching game in scenario 1. Comparison of the association obtained from (a) BDAA, (b) a global optimum for Gaussian costs with variance $\sigma=0.2$, (c) a global optimum without costs.}
 \label{fig:BDAAvsothers}
\end{figure}

Focusing on the first scenario (Figure~\ref{fig:scenarios}~(a)), we consider the three matchings shown in Figure~\ref{fig:BDAAvsothers}. The first one (a) is the stable matching resulting from the mechanism (including BDAA and Gaussian costs); The second matching (b) maximizes the sum of modified 
throughputs (i.e. including Gaussian costs); The third matching (c) maximizes the sum of unmodified throughputs (i.e. without costs).

We first observe that the proposed mechanism induces a stable matching with a drastic reduction of the unemployment problem w.r.t. the result of Figure~\ref{fig:nocost}. The natural incentives of the system resulting in a one-to-one matching have been countered and a many-to-one matching is obtained. 
The unemployment has been reduced from $73\%$ to $5\%$ in this particular scenario. 
The second point to be raised is that the proposed mechanism allows to obtain (with a polynomial complexity) a stable matching with a high modified total throughput, close to the optimal modified total throughput that is however not stable. 
For this scenario, we achieve through our mechanism  $99\%$ of the total modified maximum throughput (see Figure~\ref{fig:BDAAvsothers} (b)). 
This means that the cost for stability is very small in this particular scenario.
Furthermore, the total throughput performance of the system at the MAC layer (i.e. unmodified throughputs obtained in block MAC of the block diagram representation of the mechanism, see Figure~\ref{fig:generalblockdiagram}) is $97\%$ the total unmodified maximum throughput (see Figure~\ref{fig:BDAAvsothers} (b)) and
$47\%$ of the total maximum throughput of the uncontrolled system (see Figure~\ref{fig:BDAAvsothers} (c)). 
This quantifies the cost for control, stability and low unemployment in this scenario. 
The third point is that the quotas have been enforced by the mechanism (via the cost function) since the quotas vector from the load balancing is $\mathbf{\hat{q}} = (8.0, 4.5, 3.33, 3.83, 4.33)$ (obtained by Nash bargaining\footnote{Achieves a proportional fair allocation in the utility space of the APs. Induces the number of players to be connected to each AP.} over the share $[0,1]$ of the players at the intersection of the coverages of the APs) and the formed coalitions are of sizes 8, 4, 3, 4 and 4.\\ 

We go into more details on the difference between the quotas vector and the integer-sized coalitions in the stable matching. 
Focusing on AP3 with the quota $q_{3}=3.33$, one may observe that in case of a Gaussian cost function with unit variance, the condition for an integer quotas $3$ is only satisfied for sizes of coalitions superior or equal to $4$. This meaning that the use of a gaussian cost function centered on $3.33$ and unit variance even though increasing the penalty with the distance in sizes to $\hat{q}_{f}$ cannot guarantee the systematic incentive to form coalitions of size 3 with AP3. There exists some coalitions of size $2$ giving the players more individual throughputs than the worst coalition of size $3$. In such  case, the players will have the incentive to form the coalition with the highest individual value among those of cardinalities $2$ and $3$.
\begin{figure}[t]
\begin{center}
  \subfigure[Unemployment rates]{\epsfig{figure=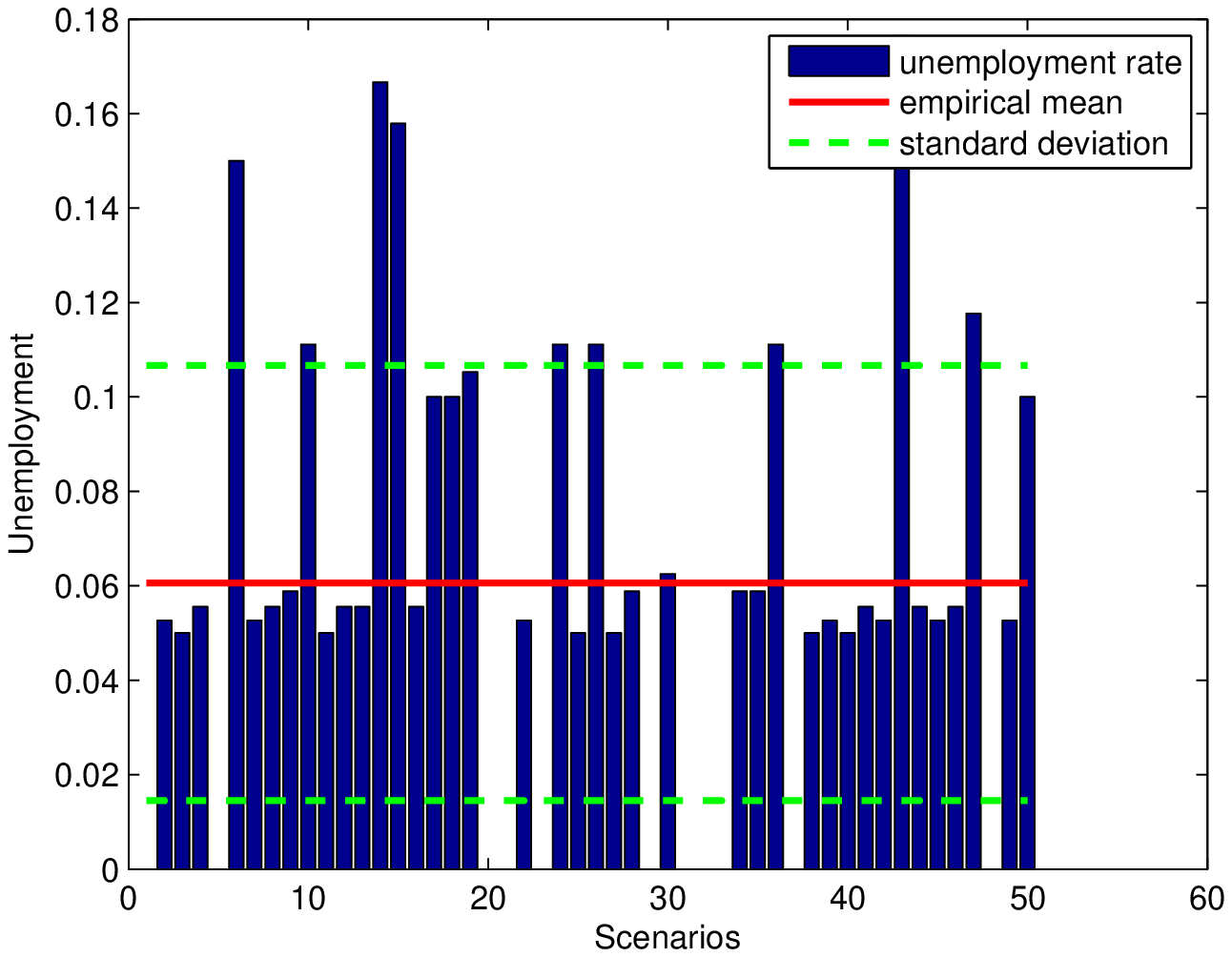,width=5cm}}
  \subfigure[Modified social welfares]{\epsfig{figure=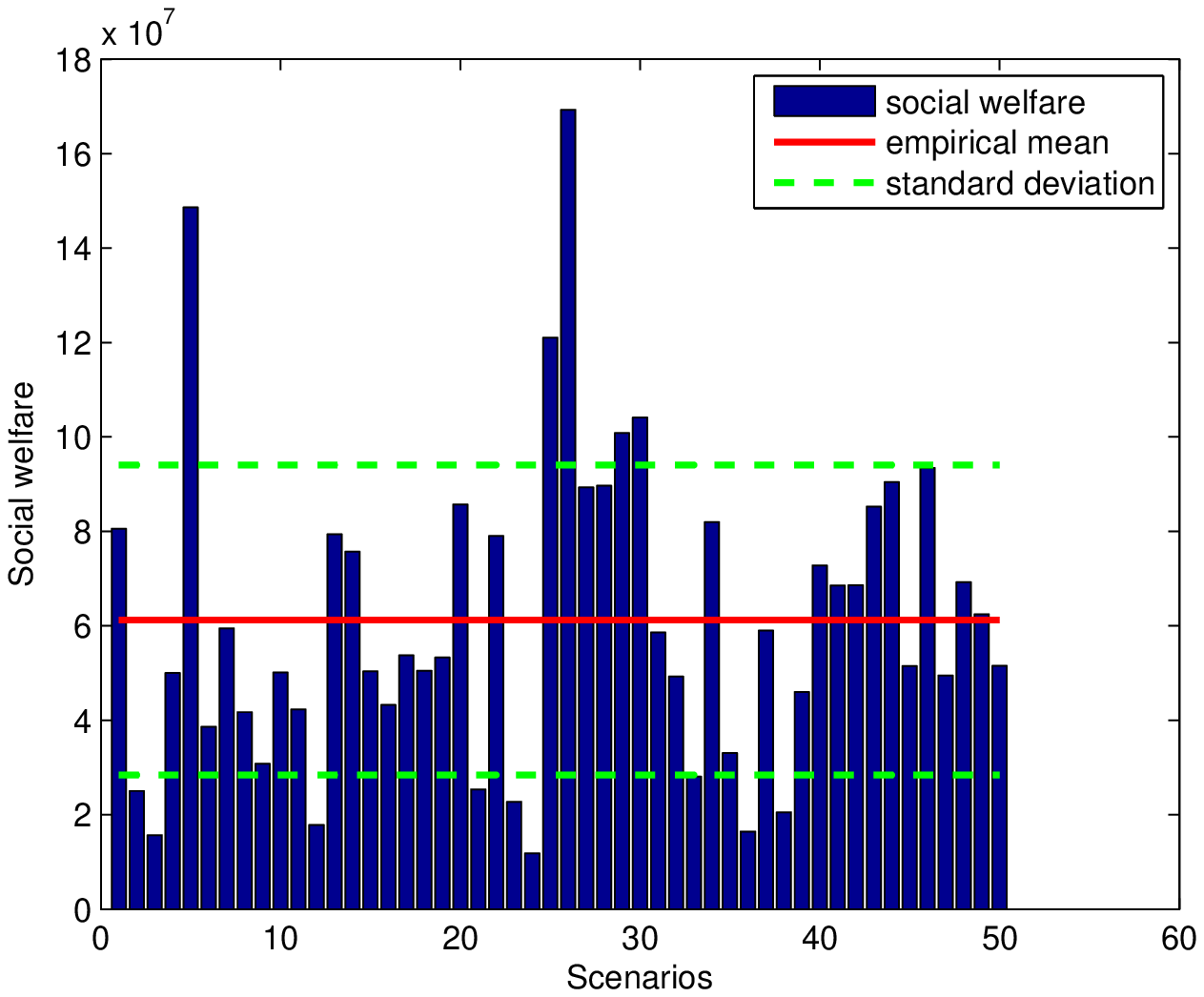,width=5cm}}
  \subfigure[Computation time of BDAA]{\epsfig{figure=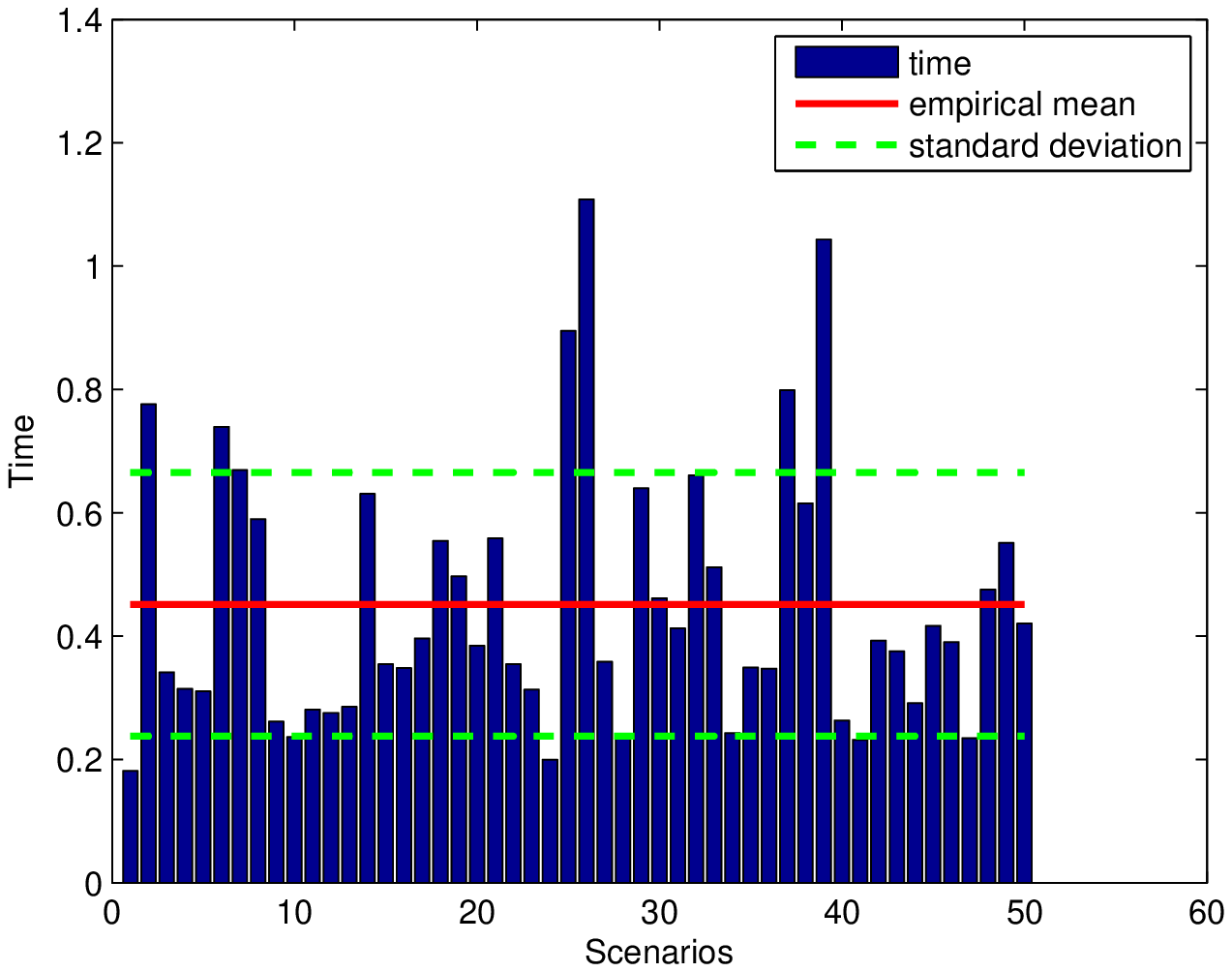,width=5cm}}
\end{center}
 \caption{(a)Unemployment rates, (b)social welfares (the social welfare of a matching is measured as the total throughput of the system at equilibrium) and (c)computation times of BDAA over a sample of 50 scenarios obtained by spatial random uniform distribution of the mobile devices. APs are spatially distributed as in Scenario 1. For each plot, the red line gives the empirical mean $\hat{m}$ of the sample and the green dotted lines the interval $[\hat{m}-\sigma, \hat{m}+\sigma]$ where $\hat{m}$ is the empirical mean of sample and $\sigma$ is the standard deviation.}
 \label{fig:AveragesRandomPositionsMobiles}
\end{figure}
\begin{figure}[h]
\begin{center}
  \subfigure[Unemployment rates]{\epsfig{figure=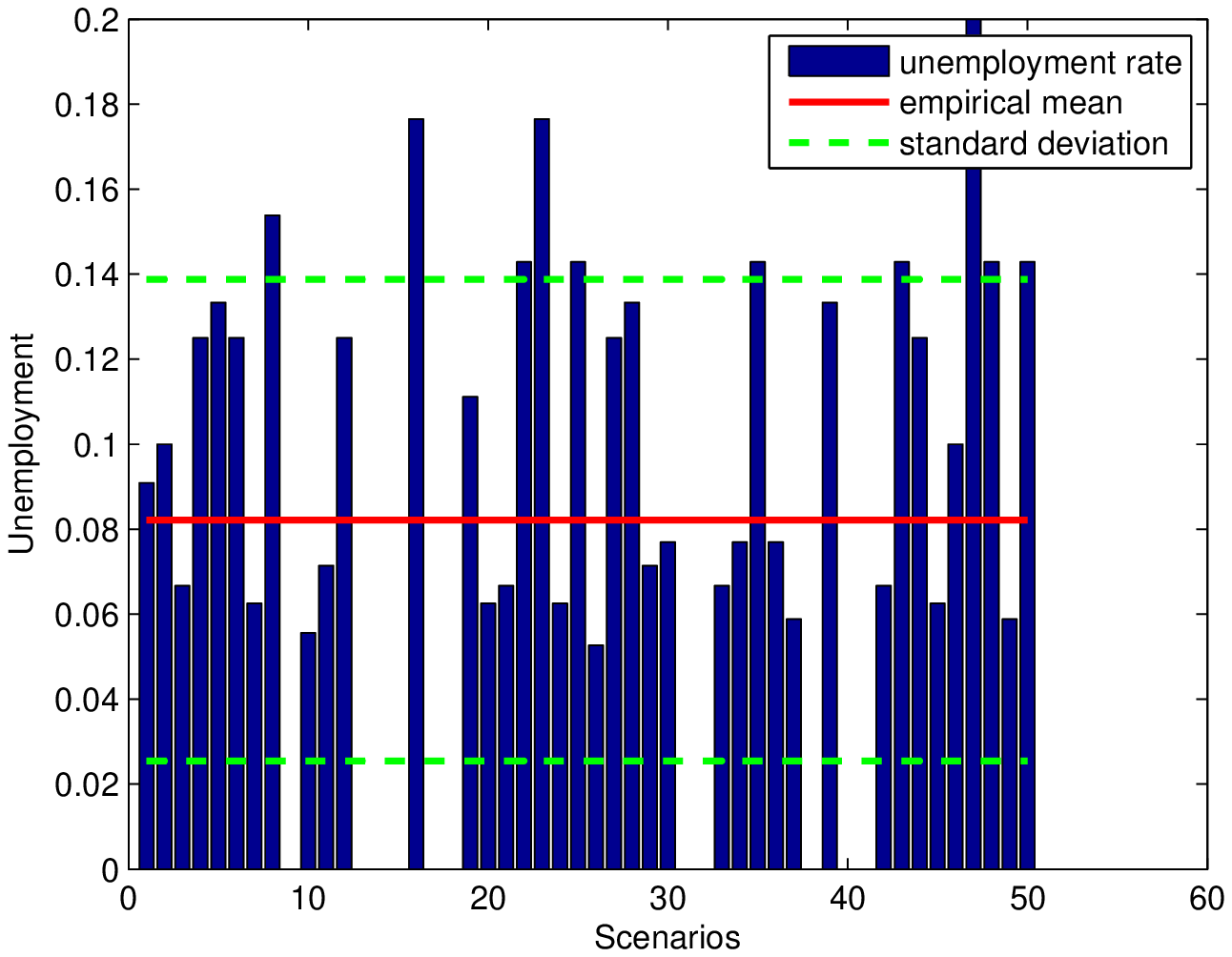,width=5cm}}
  \subfigure[Modified social welfares]{\epsfig{figure=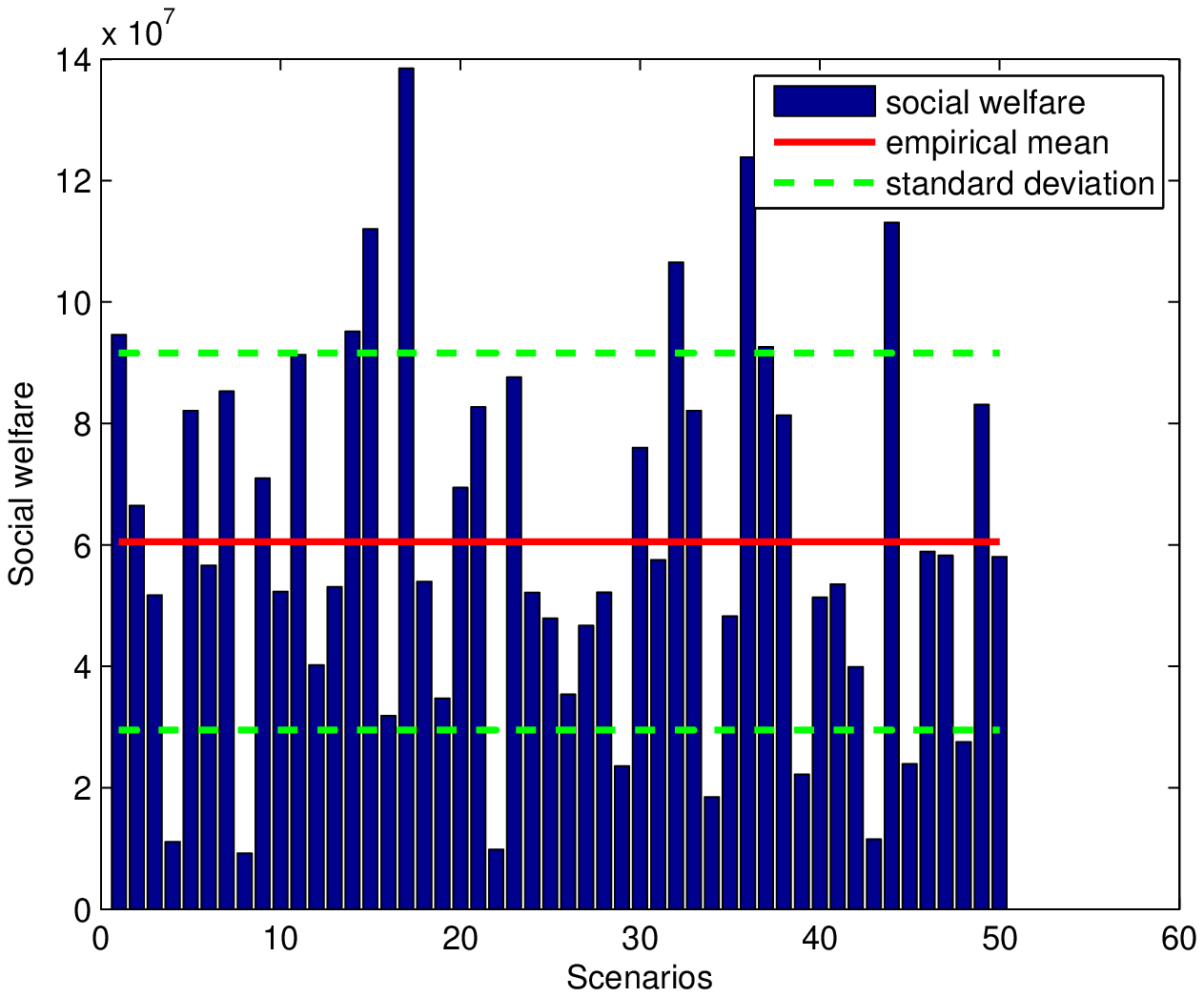,width=5cm}}
  \subfigure[Computation times of BDAA]{\epsfig{figure=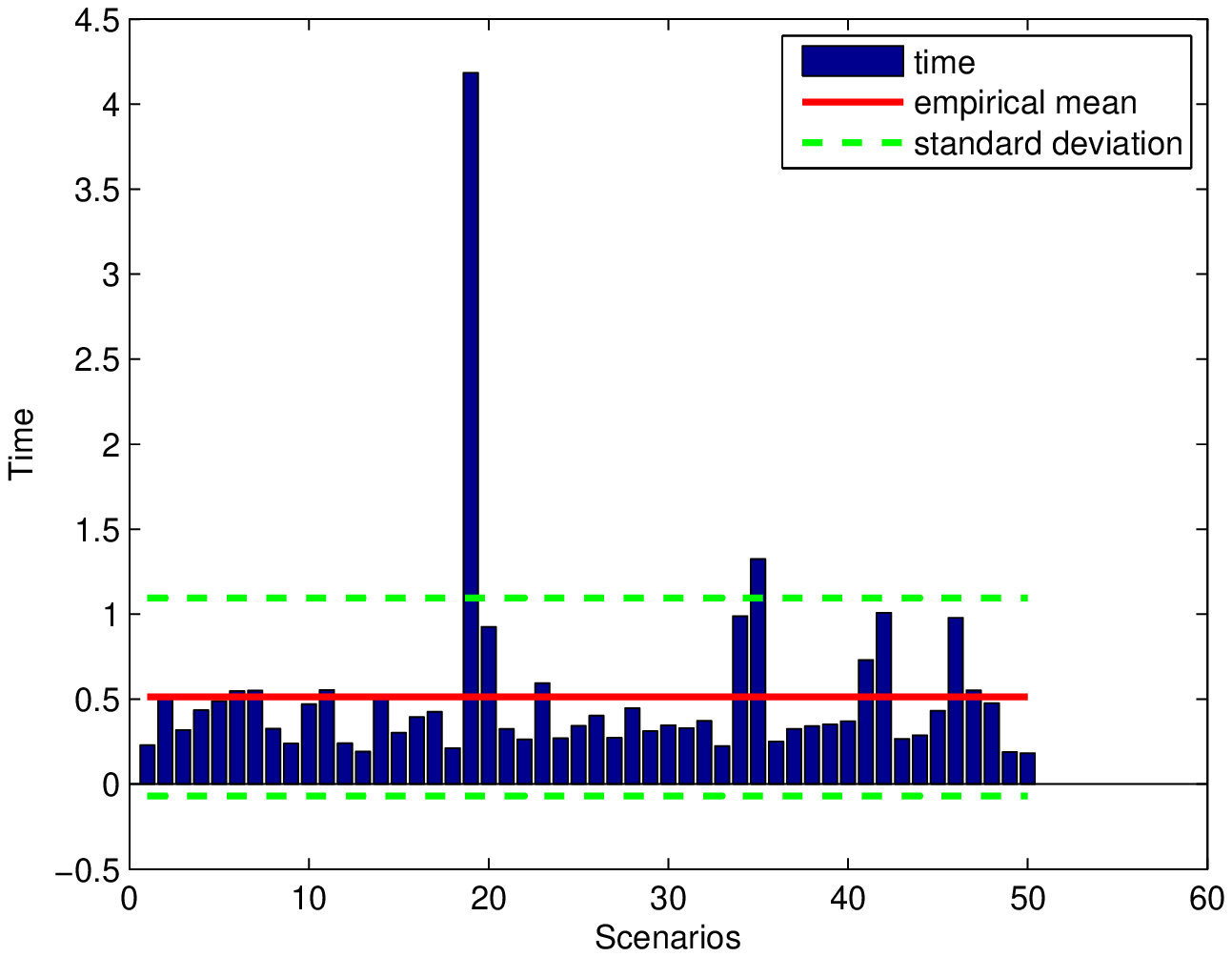,width=5cm}}
\end{center}
 \caption{(a)Unemployment rates, (b)social welfares and (c)computation times of BDAA over a sample of 50 random networks obtained by spatial random uniform distribution of the mobile devices and APs. For each plot, the red line gives the empirical mean $\hat{m}$ of the sample and the green dotted lines the interval $[\hat{m}-\sigma, \hat{m}+\sigma]$ where $\hat{m}$ is the empirical mean of sample and $\sigma$ is the standard deviation.}
 \label{fig:AveragesRandomNetworks}
\end{figure}
In Figure \ref{fig:AveragesRandomPositionsMobiles}, we show the performance of the mechanism over a set of 50 scenarios generated by spatial random uniform distribution of the mobile devices. The APs are spatially distributed as in Scenario 1 (see Figure \ref{fig:AveragesRandomNetworks} for a random distribution of both the mobile devices and APs). The red line shows the empirical mean of the sample and the green dotted lines show the interval  $[\hat{m}-\sigma, \hat{m}+\sigma]$ where $\hat{m}$ is the empirical mean of sample and $\sigma$ is the standard deviation. The empirical mean of the unemployment rate is 6\%, the mean modified social welfare is 61Mbits/s and the mean computation time of BDAA is 0.45s. Observe that in 22\% of the realizations the unemployment is null and that in 70\% of the realizations it is below the mean. 
In terms of computation times of BDAA, the mean performance is reasonably low (0.45s to match 20 mobiles to 5 APs) and the algorithm performs even better in 62\% of the scenarios.

In Figure \ref{fig:AveragesRandomNetworks}, we show the performance of the mechanism over a set of 50 scenarios generated by spatial random uniform distribution of the mobile devices and APs. The red lines show the empirical mean of the sample and the green dotted lines show the interval  $[\hat{m}-\sigma, \hat{m}+\sigma]$ where $\hat{m}$ is the empirical mean of sample and $\sigma$ is the standard deviation. The empirical mean of the unemployment rate is 8\%, the mean modified social welfare is 60Mbits/s and the mean computation time of BDAA is 0.51s. Observe that in 22\% of the realizations the unemployment is null and that in 56\% of the realizations it is below the mean. In terms of computation times of BDAA, the mean is higher than in the previous case but the algorithm performs better than the mean in 68\% of the scenarios.

\begin{figure}[h]
\begin{center}
  \subfigure[Ratios of modified performances]{\epsfig{figure=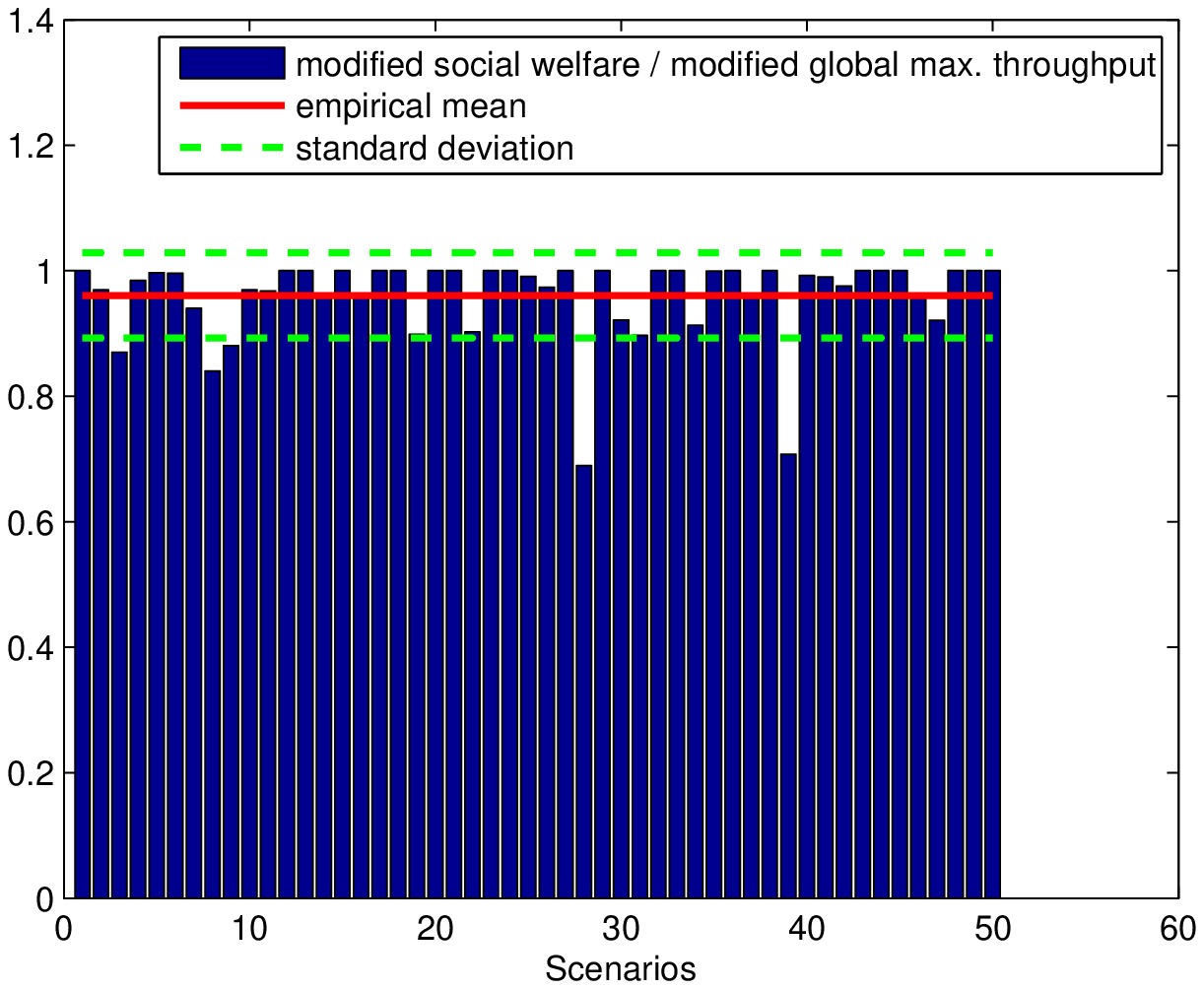,width=7cm, height=5cm}}
  \subfigure[Ratios of unmodified performances]{\epsfig{figure=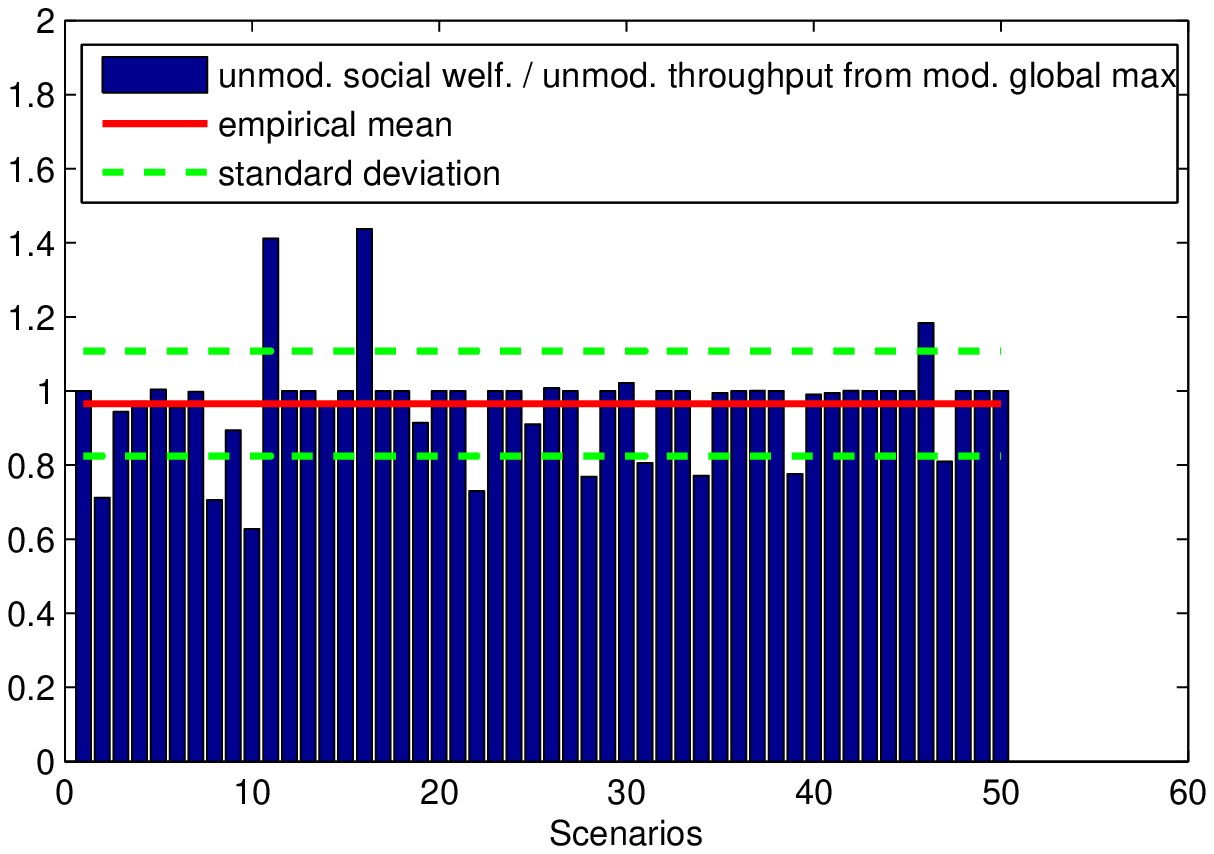,width=7cm, height=5cm}}
\end{center}
 \caption{(a)Ratios of the modified social welfares to the maximum modified (mechanism level) total throughput, (b)Ratios of the unmodified (MAC level) social welfares to the unmodified total throughputs corresponding to the matching with maximum modified total throughput. Sample of 50 random networks obtained by spatial random uniform distribution of the mobile devices and APs. For each plot, the red line gives the empirical mean $\hat{m}$ of the sample and the green dotted lines the interval $[\hat{m}-\sigma, \hat{m}+\sigma]$ where $\hat{m}$ is the empirical mean of sample and $\sigma$ is the standard deviation.}
 \label{fig:AverageRatiosRandomNetworks}
\end{figure}
In Figure \ref{fig:AverageRatiosRandomNetworks}, Plot (a), we show the ratios of the modified (mechanism level) social welfare (by definition, the total throughput resulting from BDAA) to the maximum total modified throughput. 
The mean performance of BDAA achieves 96\% of this maximum.  Furthermore, observe that the global maximum is achieved by BDAA in 46\% of the random networks. Finally, the ratio is below $\hat{m}-\sigma$ in only 10\% of the cases.
In Figure \ref{fig:AverageRatiosRandomNetworks}, Plot (b), we show the ratios of the unmodified (MAC level) social welfare to the unmodified total throughput induced at the matching maximizing the total modified throughput. 
The mean performance of BDAA achieves 97\% of this unmodified total throughput. 
Finally, observe that in some cases, the ratio is even higher than one. This means that BDAA gives a total throughput at the MAC level that is superior to the (unmodified) total throughput resulting from the maximization at the mechanism level (modified values) while the ratio was inferior to one in the modified case.
This may come from the fact that in some cases, the equilibrium point (stable matching resulting from BDAA) may contain coalitions with lower modified worths (because of the penalization) w.r.t. those in the global maximum but higher worths at MAC level (real unpenalized setting). 
\begin{figure}[t]
\begin{center}
  
      \subfigure[Stable matching resulting from Gaussian cost and BDAA.]{\epsfig{figure=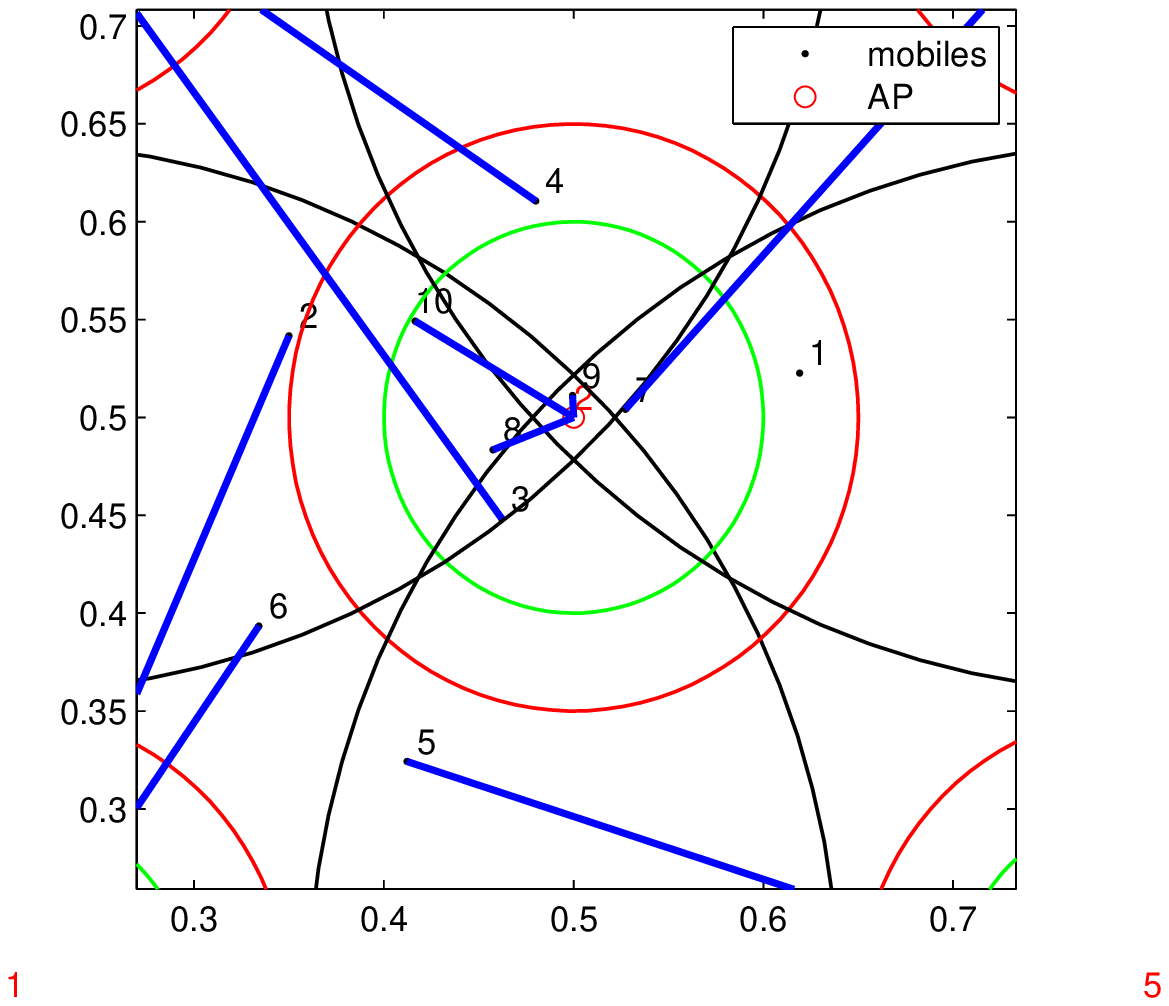,width=6cm}}
      \subfigure[Matching resulting from the best-RSSI scheme.]{\epsfig{figure=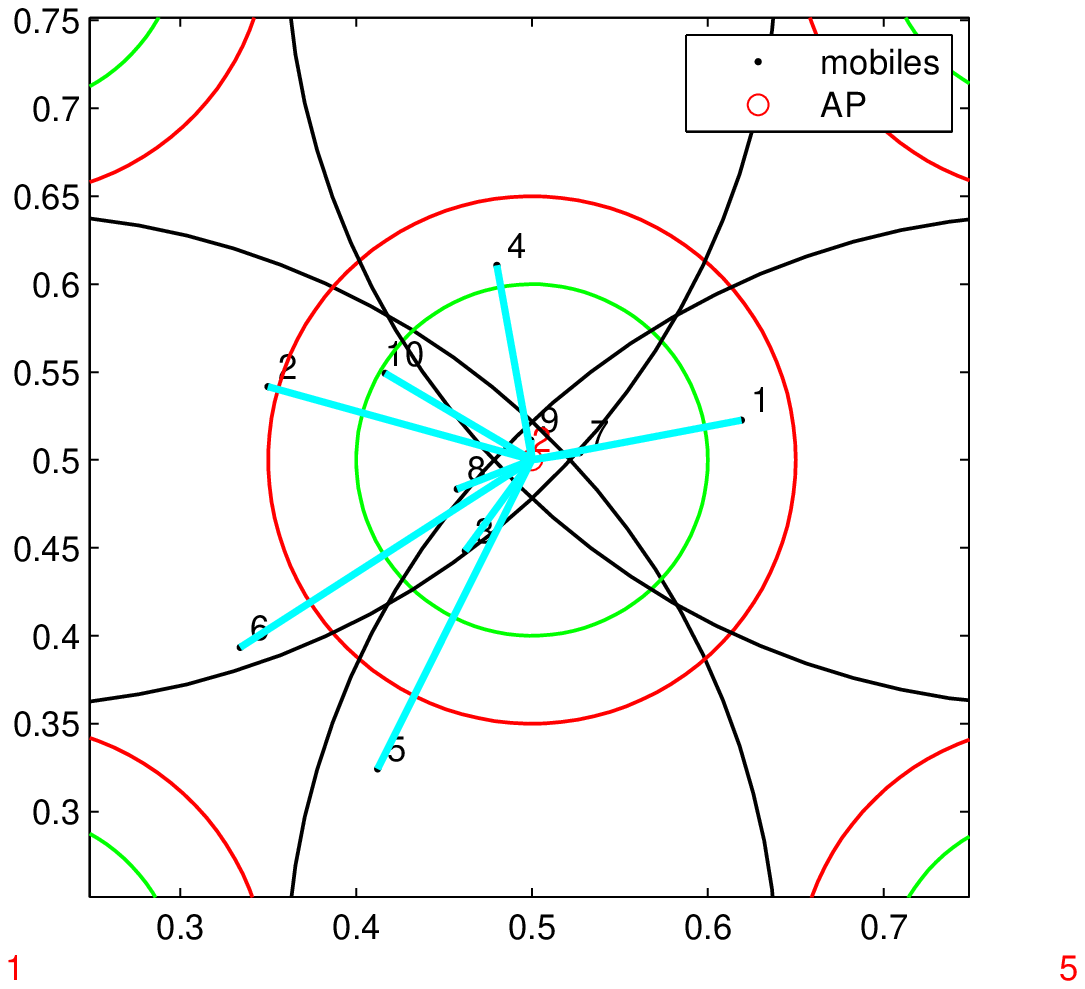,width=6cm}}
\end{center}
 \caption{Comparison of the association obtained from (a) BDAA and (b) the best-RSSI scheme in scenario 2. These two figures show the AP at center (zoom).}
 \label{fig:BDAAvsRSSI}
\end{figure}

We now compare our approach to the best-RSSI scheme in a scenario with high congestion (Scenario 2). The two matchings are compared Figure \ref{fig:BDAAvsRSSI}. We observe that the load is effectively shared among the APs and that the individual throughputs are greatly increased from $527$~kbits/s when using best-RSSI to $1.64$~Mbits/s for the coalition with AP1, $1.93$~Mbits/s for the coalition with AP2, $2.59$~ Mbits/s for the coalition with AP3, $1.64$~Mbits/s for the coalition with AP4 and $2.59$~Mbits/s for the coalition with AP5. The individual performances are multiplied by a factor 3 to 5 in this particular scenario.

\section{Conclusion}
\label{sec:Conclusion}
In this paper, we have presented a novel AP association mechanism in multi-rate IEEE 802.11 WLANs. We have formulated the problem as a coalition matching game with complementarities and peer effects and we have provided a new practical control mechanism that provides nodes the incentive to form coalitions both solving the unemployment problem and reducing the impact of the anomaly in IEEE 802.11.
Simulation results have shown that the proposed mechanism can provide significant gains in terms of increased throughput by minimizing the impact of the anomaly through the overlapping between APs. We have also proposed a polynomial complexity algorithm for computing a stable structure in many-to-one matching games with complementarities and peer effects. This work is a first step in the field of controlled coalition games for achieving core stable associations in distributed wireless networks. Further works includes for example the study of a dynamic number of users or the impact of interference. 

\clearpage
\appendices

\section{IEEE 802.11 as a Nash Bargaining} \label{app:80211-NB}

In this appendix, we show that, in any cell, the resource allocation (in throughput) induced by the IEEE 802.11 protocol in DCF implementation can be formulated as a Nash bargaining.
In \cite{Kum2007}, Altman et a.l. assume the single-cell network hypothesis \footnote{"several IEEE 802.11 compliant players within such a distance of each other that only one transmission can be sustained at any point of time"}, the decoupling approximation and the saturation of queues.
They provide an analytical formula for the saturation throughput $r_{i,j}$ (bits per slot, see equation \eqref{eq:flowsaturationthroughput}) of any flow $(i,j)$ in a cell $C$,
\begin{equation}\label{eq:flowsaturationthroughput}
r_{ij,C} = \frac{p_{ij}L_{ij}\beta(1-\beta)^{|C|}}{1+\sum\limits_{i=1}^{|C|}(\beta(1-\beta)^{|C|-1}((\sum\limits_{k=1}^{m_{i}}p_{ik}\frac{L_{ik}}{\theta_{i,k}})+T_{0}))+((1-(1-\beta)^{|C|}-|C|\beta(1-\beta)^{|C|-1})T_{C})}
\end{equation}
where $T_{o}$ is the fixed overhead with a packet transmission in slots, $T_{c}$  is the fixed overhead for an RTS collision in slots, $\beta$ is the attempt rate of any player and $L_{ij}$ is the packet length of flow $(i,j)$.  
Mobile device i's saturation throughput is thus obtained by summing over i's $m_{i}$ flows as,
\begin{equation}
r_{i,C} = \sum\limits_{j=1}^{m_{i}}
	\frac{p_{ij}L_{ij}\beta(1-\beta)^{|C|}}{1+\sum\limits_{i=1}^{|C|}(\beta(1-\beta)^{|C|-1}((\sum\limits_{k=1}^{m_{i}}p_{ik}\frac{L_{ik}}{\theta_{i,k}})+T_{0}))+((1-(1-\beta)^{|C|}-|C|\beta(1-\beta)^{|C|-1})T_{C})}
\end{equation}
which can equivalently be written as,
\begin{equation}\label{eq:sharingRuleWifi}
	r_{i,C} = \frac{\alpha_{i}}{\sum\limits_{k\in C}\alpha_{k}}v(C)
\end{equation}
where, for any mobile device $k$ in $C$,
\begin{equation}
	\alpha_{k} = \sum\limits_{j=1}^{m_{k}}p_{kj}\frac{L_{kj}}{L_{max}} \leq 1
\end{equation}
and,
\begin{equation}
	v(C) = \sum\limits_{i\in C}\sum\limits_{j=1}^{m_{i}}p_{ij}L_{ij}\frac{\beta(1-\beta)^{|C|}}{1+\sum\limits_{i=1}^{|C|}(\beta(1-\beta)^{|C|-1}((\sum\limits_{k=1}^{m_{i}}p_{ik}\frac{L_{ik}}{\theta_{i,k}})+T_{0}))+((1-(1-\beta)^{|C|}-|C|\beta(1-\beta)^{|C|-1})T_{C})}
\end{equation}
This allocation is Nash solution to the bargaining problem (see Section \ref{subsec:BNB} for more details) modeling the competition (with null threats) over $v(C)$ between the mobile devices of a set $C$ with utility functions, 
\begin{equation}\label{eq:utilities GeneralSaturated80211}
	u_{i}(x) = x^{\alpha_{i}}, \quad \forall i\in \mathcal{N}
\end{equation}


\section{Remarks on the existence of core stable structures} \label{app:regularity}

In Proposition \ref{prop:stability}, it is required that the number $F$ of APs is superior or equal to $2$ and that the coalitions of the set $\mathcal{C}$ satisfy some conditions in terms of maximum cardinalities (for any $f$ in $\mathcal{F}$, $q_{f}\in\{2,\ldots,W-1\}$). 
Under such conditions, it can be sown that the set of coalitions $\mathcal{C}$ is regular which is a necessary condition for the existence of core stable structures (see \cite{Pyc2012}, Corollary 2).
\begin{definition}[Regularity of a set of coalitions \cite{Pyc2012}]\label{def:regularCoalitions}
A set of coalitions is regular if there is a partition of the set of agents $\mathcal{N}$ into two disjoint, possibly empty, subsets $\mathcal{F}$ and $\mathcal{W}$  that satisfy the following three assumptions: \\
{\bf C1.} For any two different players, there exists a coalition containing them if and only if at least one of the players is a player of $\mathcal{W}$. \\
{\bf C2.} For any players $a_{1}$, $a_{2}\in\mathcal{W}$ and player $a_{3}$, there exist proper ($\neq \mathcal{N}$) coalitions $C_{1,2},C_{2,3},C_{3,1}$ such that $a_{k},a_{k+1}\in C_{k,k+1}$ and $C_{1,2}\cap C_{2,3} \cap C_{3,1}\neq \emptyset$. \\
{\bf C3.} (i) For any player $w\in\mathcal{W}$ and player $a$, if $\{a,w\}$ is not a coalition then there are two different players $f_{1},f_{2}\in\mathcal{F}$ such that $\{f_{1},a,w\}$ and  $\{f_{2},a,w\}$ are coalitions. (ii) No coalition, which is different from $\mathcal{N}$ contains $\mathcal{W}$.
\end{definition}
As already observed, the following result can be shown,
\begin{lemma}[\cite{Pyc2012}] \label{lemma:regularity}
	The set of coalitions $\mathcal{C}$ defined in (\ref{eq:coalitionC}) is regular if $q_{f}\in\{2,\ldots,W-1\}$ and $F\geq 2$.
\end{lemma} 
Furthermore, it is also required that the sharing rule $D$ is such that the functions $D_{i,C}$ (mapping the worth of $C$ to $i$'s share in $C$) are strictly increasing, continuous and go to $+\infty$ as the worth of the coalition goes to $+\infty$. In fact, there are also necessary condition on the sharing rule for the existence of core stable structures (see \cite{Pyc2012}, Corollary 2).
Such conditions are called regularity conditions of the sharing rule. Before entering the details of these , we need to define the pairwise alignment of preferences,
\begin{definition}[Pairwise aligned preferences\cite{Pyc2012}]\label{def:pairwiseAlignment}
	Preferences are pairwise aligned if for all agents $a,b\in\mathcal{N}$ and proper coalitions $C,C'$ that contain $a,b$, we have: $C\succeq_{a} C' \Leftrightarrow C\succeq_{b} C'$.
\end{definition}
We now give the definition of regular sharing rules,
\begin{definition}[Regularity of a sharing rule \cite{Pyc2012}]\label{def:regularSharingRules}
	A sharing rule $D$ is regular if:
	\begin{itemize}
		\item (i) It is pairwise aligned over the grand coalition: either $\mathcal{N}\not\in \mathcal{C}$, or $\mathcal{N}\in \mathcal{C}$ and the pairwise alignment equivalence is true whenever $C$ or $C'$ equals $\mathcal{N}$
		\item (ii) All functions $D_{i,C}$ are strictly increasing, continuous and $\lim_{y\rightarrow +\infty} D_{i,C}(y) = +\infty$
	\end{itemize}
\end{definition}
The regularity conditions of the sharing rule are satisfied when considering the resource allocation resulting from the use of the IEEE 802.11 protocol in DCF implementation (see Appendix \ref{app:80211-NB}). In fact, (i) is verified because the grand coalition $\mathcal{N}$ is not in $\mathcal{C}$ and (ii) is verified by linearity of the WiFi allocation (see equation \eqref{eq:sharingRuleWifi}) in the worth $v(C)$.\\

\section{Two examples of control of incentives in games} \label{app:examples-control}

We give two simple examples of the mechanism we propose to control the players' individual incentives. These examples do not give technical details but illustrate the underlying motivations of the work developed in the framework of the wireless technologies.

\begin{example}
	Consider the case of a many-to-one matching game with ordinal preferences $\Gamma = (\mathcal{W},\mathcal{F},\mathbf{P})$, where $\mathbf{P}$ is the list of stated preferences.
	This description of matching games does not allow for the game theorists to catch the strategic incentives the player is facing when emitting preferences since it only shows the emitted preferences. It does not even allow to exhibit the set of players' feasible strategies.
	As a consequence of this incompleteness and in addition to the standard form, matching games have been defined in a strategic form $\Gamma = (\mathcal{W},\mathcal{F},\{Q_{i}\},h,\mathbf{P})$, where $\{Q_{i}\}$ is the set of players' individual feasible strategies, $h$ is the matching mechanism and $\mathbf{P}$ is the set of true preferences.
	For the ease of understanding, Roth et a.l. assume in \cite{Rot1990} that the set of feasible individual strategies is reduced to the set of preferences lists the players' may state. This description allows for each player to state a list of preferences different from his true preferences in view of manipulating the association mechanism (or matching mechanism).
	In the framework of matching games in wireless networks, we do not allow for the players (devices) to manipulate their preferences or behave strategically by misstating since, we assume that they would intrinsically be built and programmed so as to respect a given protocol asking them for truthfulness.
	Nevertheless, this framework can be used to allow the controller to modify the game so as to change the true (stated preferences by the previous assumptions) preferences.
	Thus, the incitation operator $\Omega$ operates on the matching game in strategic form $\Gamma = (\mathcal{W},\mathcal{F},\{Q_{i}\},h,\mathbf{P})$, in a way that changes the true and stated preferences such that,
		\begin{equation}
			\Gamma' = (\mathcal{W},\mathcal{F},\{Q_{i}\},h,\mathbf{P'}) = \Omega(\Gamma)
		\end{equation}
\end{example}

The second example is given in terms of the firms and workers model commonly used in stable matchings \cite{Rot1990}.
\begin{example}
Assume a set of firms and a set of workers. Each firm can hire workers. We call a coalition a subset of players containing a single firm and one or more workers. The characteristic function assigns each coalition the worth it produces and this worth is shared via a Nash bargaining among the players in each coalition. The players have the incentive to form coalitions maximizing their own payoffs. In an employment market with performances similar to the IEEE 802.11 standard, the characteristic function would be increasing in the productivity types and the individual payoffs would be sub-additive. In this case, firms and workers would have the incentive to group by pairs of highest productivity types. Nevertheless, these stable structures of coalitions leave some workers unemployed which is not satisfying from the point of view of the unemployment market. Facing the problem, a government solely interested in reducing the number of unemployed workers would thus seek for (well-designed) tax rates as levers to provide the players the incentives to form coalitions reducing the employment. Assuming Nash bargaining for the resource allocation, we propose to manipulate the payoffs (incentives for players via their individual increasing concave utilities) by the way of such tax rates applied to the gross income(s). In other words, we manipulate the individual payoffs so as to change the equilibrium point of the coalition formation process.
\end{example}

\section{Proofs} 
\label{appendix:proofs}
\subsection{Proof of Proposition~\ref{prop:controlincentives}}
\begin{proof}
	Assume a coalition game $\Gamma = (\mathcal{F}\cup\mathcal{W}, v, \{u_{i}\}_{i\in N})$ in characteristic form with the Nash bargaining sharing rule over the $v(C)$-simplex in each coalition $\mathcal{C}$. Furthermore assume increasing, three-times differentiable, and concave utility functions $u_{i}:\mathbb{R}^{+}\rightarrow\mathbb{R}^{+}, i\in\mathcal{N}$.\\ 	
	Using Theorem \ref{the:solutionnashproduitutilities}, the Nash solution to the bargaining problem in any coalition $C\in \mathcal{C}$ solves,
	\begin{equation*}
		\begin{aligned}
		& \underset{\mathbf{x}}{\text{maximize}}
		& & \prod\limits_{i\in C}u_{i}(x_{i}) \\
		& \text{subject to}
		& & \sum\limits_{i\in C}x_{i,C} \leq v(C) \\
		&&& x_{i} \geq 0 \; \forall i\in C
		\end{aligned}
	\end{equation*}
	Which can be equivalently written in the following form,
	\begin{equation*}
		\begin{aligned}
		& \underset{\mathbf{x}}{\text{minimize}}
		& & -\sum\limits_{i\in C}\log(u_{i}(x_{i})) \\
		& \text{subject to}
		& & \sum\limits_{i\in C}x_{i,C} - v(C) \leq 0\\
		&&& -x_{i} \leq 0 \; \forall i\in C
		\end{aligned}
	\end{equation*}	
	This optimization problem is convex by convexity of the $v(C)$-simplex and of the objective function.\\
	For any coalition $C\in\mathcal{C}$ such that $v(C) = 0$, solving the allocation problem is irrelevant. The player receives a null payoff.
	For any other coalition (with strictly positive worth $v(C)$), the interior of the $v(C)$-simplex ($B_C=\{\mathbf{s}_C=(s_{i,C})_{i\in C}|\sum_{i\in C}s_{i,C}\leq v(C)\}$) is non-empty.
	Using Slater's constraint qualification\footnote{Strong duality holds for a convex optimization problem if it is strictly feasible.}, strong duality holds for this convex  optimization problem.
	There is an optimal solution to the optimization problem iff the Karush-Kuhn-Tucker conditions can be satisfied.\\
	We have the Lagrangian $\mathcal{L}(\mathbf{x}, \lambda_{0}, \{\lambda_{i}\}_{i\in C})$ such that,
	\begin{multline}
		\mathcal{L}(\mathbf{x}, \lambda_{0}, \{\lambda_{i}\}_{i\in C})  =-\sum\limits_{i\in C}\log(u_{i}(x_{i}))
		 + \lambda_{0}\left(\sum\limits_{i\in C}x_{i} - v(C)\right) - \sum\limits_{i\in C}\lambda_{i}x_{i}
	\end{multline}
	We have the K.K.T. conditions such that:
	\begin{itemize}
		\item Primal constraints: (i)$\sum\limits_{i\in C}x_{i} - v(C)\leq 0$, (ii) $-x_{i}\leq 0\, \forall i\in C$
		\item Dual constraints: $\lambda_{0}\geq 0$, $\lambda_{i}\geq 0 \,\forall i\in C$
		\item Complementary slackness: (i) $\lambda_{0}\left( \sum\limits_{i\in C}x_{i} - v(C)\right) = 0$, (ii)$\lambda_{i}x_{i} = 0 \, \forall i\in C$
		\item Vanishing gradient of the Lagrangian at the solution point: $\frac{u_{i}^{'}(x_{i})}{u_{i}(x_{i})} = \lambda_{0} - \lambda_{i} = \frac{1}{\chi_{i}(x_{i})} \, \forall i\in C$
	\end{itemize}
	
%
	\textbf{Case :} $\lambda_{0} > 0$ and $\lambda_{i} = 0\, \forall i\in C$\\
	Due to the complementary slackness conditions, we must have	$\sum\limits_{i\in C}x_{i} = v(C)$ and $x_{i}\geq 0 \, \forall i\in C$.\\
	The vanishing gradient of the Lagrangian condition gives, $\forall i\in C$:
	\begin{equation}
		\frac{u_{i}^{'}(x_{i})}{u_{i}(x_{i})} = \lambda_{0} = \frac{1}{\chi_{i}(x_{i})} = \frac{1}{\chi_{C}}
	\end{equation}
	Knowing that $u_{i}$ is concave and strictly increasing, we have $\left(\frac{u_{i}^{'}(x_{i})}{u_{i}}\right)' = \frac{u_{i}''u_{i} - (u_{i}')^{2}}{u_{i}^{2}} < 0$. 
	Thus, it is a strictly monotonic function and it admits an inverse.
	We denote $\left(\frac{u_{i}^{'}}{u_{i}}\right)^{-1}$ this inverse.\\
	We have:
	\begin{equation}\label{eq:optimalallocation}
		x_{i} = \left(\frac{u_{i}^{'}}{u_{i}}\right)^{-1}(\lambda_{0}) \quad \forall i\in C
	\end{equation}
	Due to the complementary slackness conditions, we must have:
	\begin{equation}
		\sum\limits_{i\in C}\left(\frac{u_{i}^{'}}{u_{i}}\right)^{-1}(\lambda_{0}) = v(C)
	\end{equation}
	The function $\sum\limits_{i\in C}\left(\frac{u_{i}^{'}}{u_{i}}\right)^{-1}$ is also strictly monotonic and has an inverse on $\mathbb{R}^{+*}$.
	We denote it $\left(\sum\limits_{i\in C}\left(\frac{u_{i}^{'}}{u_{i}}\right)^{-1}\right)^{-1}$.
	The optimal Lagrange multiplier is obtained as:
	\begin{equation}\label{eq:optimalmultiplier}
		\lambda_{0} = \left(\sum\limits_{i\in C}\left(\frac{u_{i}^{'}}{u_{i}}\right)^{-1}\right)^{-1}(v(C))
	\end{equation}
	We thus have the optimal solution of the Nash bargaining problem by solving \eqref{eq:optimalallocation} and \eqref{eq:optimalmultiplier}.
	In terms of the fear-of-ruin $\chi_{C}$:
	\begin{align}
		&x_{i} = \left(\frac{u_{i}^{'}}{u_{i}}\right)^{-1}(\frac{1}{\chi_{C}})\label{eq:optimalsolfromfor}\\
		&\chi_{C} = \frac{1}{\left(\sum\limits_{i\in C}\left(\frac{u_{i}^{'}}{u_{i}}\right)^{-1}\right)^{-1}(v(C))}\label{eq:foroptimal}
	\end{align}
	We now turn to the analysis of the function $b_{i} \triangleq \frac{u_{i}^{'}}{u_{i}}$ called boldness of player $i$.
	We have:
	\begin{equation}
		b_{i}^{'}(x_{i}) = \frac{u_{i}^{''}(x_{i})u_{i}(x_{i}) - \left(u_{i}^{'}(x_{i})\right)^{2}}{\left(u_{i}(x_{i})\right)^{2}}
	\end{equation}
	By assumption, $u_{i}$ is strictly increasing and concave for any player $i$.
	For any player $i$, we  obtain that $b_{i}^{'}(x_{i})$ is strictly negative for any $x_{i}$ and thus that the boldness $b_{i}$ is a decreasing function of $x_{i}$.\\
	Thus, its inverse $\left(\frac{u_{i}^{'}}{u_{i}}\right)^{-1}$ is also a decreasing function and the fear-of-ruin of player $i$, $\frac{u_{i}}{u_{i}'}$ is an increasing function of $x_{i}$.\\
	The sum of decreasing functions, $\sum\limits_{i\in C} \left(\frac{u_{i}^{'}}{u_{i}}\right)^{-1}$ is a decreasing function. So is its inverse $\left(\sum\limits_{i\in C} \left(\frac{u_{i}^{'}}{u_{i}}\right)^{-1}\right)^{-1}$.\\
	As a consequence, we obtain that the common boldness $\lambda_{0}$ (solving the Nash bargaining optimization problem) is a decreasing function of the common wealth $v(C)$ and thus (using $\chi_{C} = \frac{1}{\lambda_{0}}$) that the common fear-of-ruin is an increasing function of the common wealth $v(C)$.\\
	Finally, using \eqref{eq:optimalallocation}, we obtain that $x_{i}$ is decreasing in $\lambda_{0}$ but increasing in $v(C)$ for each player $i$ in $C$.
	It is an increasing function of the fear-of-ruin $\chi_{C}$ (by \eqref{eq:optimalsolfromfor}).\\
	Now, assume two coalitions $C$ and $C'$ and their respective Nash solutions to the bargaining problem $\mathbf{x_{C}}$ (where the bold notation $\mathbf{x_{C}}$ denotes the vector of individual allocations solving the Nash bargaining optimization program in coalition $C$) and $\mathbf{x_{C'}}$.\\
	If we want all the players in $C$ and $C'$ (i.e. in $C\cap C'$) to prefer $C$ to $C'$, then we must have:
	\begin{equation}\label{eq:inequalitites}
		x_{i,C} > x_{i,C'} \quad \forall i\in C\cap C' 
	\end{equation}
	which can equivalently be written:
	\begin{equation}
		(x_{i,C})_{i\in C\cap C'} \succ (x_{i,C'})_{i\in C\cap C'}
	\end{equation}
	where $\succ$ denotes the component wise strict inequality in $\mathbb{R}^{|C\cap C'|}$.\\
	The number of players in $C\cap C'$ can be arbitrary large (depending on $C$ and $C'$).
	So is the number of inequalities of the form of \eqref{eq:inequalitites}  that must be satisfied.
	This number is of order $O(N)$.
	In order to reduce the complexity of the control of the incentives of an order $N$, we use the fact that the Nash solution to the bargaining problem is component-wisely increasing in a quantity that is constant over the players in the coalition, namely the fear-of-ruin.\\
	As a consequence, the set of inequalities in \eqref{eq:inequalitites} can be reduced to the following scalar inequality:
	\begin{equation}
		\chi_{C} > \chi_{C'}
	\end{equation}
	which can be written as:
	\begin{equation}
		 \frac{1}{\left(\sum\limits_{i\in C}\left(\frac{u_{i}^{'}}{u_{i}}\right)^{-1}\right)^{-1}(\tilde{v}(C))}
		 >
		 \frac{1}{\left(\sum\limits_{i\in C'}\left(\frac{u_{i}^{'}}{u_{i}}\right)^{-1}\right)^{-1}(\tilde{v}(C'))}
	\end{equation}
	where $\tilde{v}$ is a characteristic function.\\
	Taking the inverse, we obtain:
	\begin{equation}
		\left(\sum\limits_{i\in C}\left(\frac{u_{i}^{'}}{u_{i}}\right)^{-1}\right)^{-1}(\tilde{v}(C))
		<
		\left(\sum\limits_{i\in C'}\left(\frac{u_{i}^{'}}{u_{i}}\right)^{-1}\right)^{-1}(\tilde{v}(C'))
	\end{equation}
	Denoting $F_{C} = \left(\sum\limits_{i\in C'}\left(\frac{u_{i}^{'}}{u_{i}}\right)^{-1}\right)^{-1}$, we can write:
	\begin{equation}
		F_{C}\circ \tilde{v}(C) < F_{C'}\circ \tilde{v}(C')
	\end{equation}
	We thus obtain the set of transformations $\Omega$ from the set of characteristic functions in itself that provide the players the incentive for some subset of coalition $\mathcal{C'}\subset\mathcal{C}$ must satisfy the following scalar inequalities, $\forall C'\in\mathcal{C}',\forall C\in\mathcal{C}\backslash \mathcal{C}'$ s.t. $C'\cap C \neq \emptyset$:
	\begin{equation}
		F_{C'}\circ\Omega(v)(C') < F_{C}\circ\Omega(v)(C)
	\end{equation}
	where $F_{C} = \left(\sum\limits_{i\in C}\left(\frac{u_{i}^{'}}{u_{i}}\right)^{-1}\right)^{-1}$ and $v$ is the characteristic function of the original coalition game.\\
	This concludes the proof.
\end{proof}

\subsection{Proof of Corollary~\ref{cor:controlincentives}}

\begin{proof}
	Let $\mathcal{C}_{f}$ denote the set of coalitions containing the AP $f\in\mathcal{F}$.
For every AP $f\in\mathcal{F}$, we want the vector of individual payoffs to be decreasing with the distance to the objective $\hat{q}_f$ where the distance function $d:\mathcal{C}\times \mathcal{C}\rightarrow \mathbb{N}$ is defined such that $d(C,C') = ||C| - |C'||$.\\
In other words, we want any coalition of size $q$ to be strictly preferred to any coalition of size $q+1$ for any size $q$ superior or equal to the objective $\hat{q}_{f}$.
We furthermore want any coalition of size $q$ to be strictly preferred to any coalition of size $q-1$ for any size $q$ inferior or equal to the objective $\hat{q}_{f}$.
Denoting $\mathbf{u}(\mathbf{x_{C}})$ the vector of utilities of the allocation $\mathbf{x_{C}}$ solving the Nash bargaining optimization program in coalition $C$), we want:

	\begin{equation}
		\left\{
		\begin{array}{l}
			\mathop{\min}\limits_{\substack{C\in\mathcal{C}_{f}\\ s.t. |C|=q}} \mathbf{u(x_{C})} \succ \mathop{\max}\limits_{\substack{C\in \mathcal{C}_{f}\\ s.t. |C|= q+1}}\mathbf{u(x_{C})},\quad \forall {q}\geq \hat{q}_{f}
			\\
			\mathop{\min}\limits_{\substack{C\in\mathcal{C}_{f}\\ s.t. |C|=q}}\mathbf{u(x_{C})} \succ \mathop{\max}\limits_{\substack{C\in \mathcal{C}_{f}\\ s.t. |C|= q-1}}\mathbf{u(x_{C})},\quad \forall {q}\leq \hat{q}_{f}
		\end{array}
		\right.
	\end{equation}

Using the fact that the utilities are increasing functions of the payoffs, we obtain the following equivalent condition,	
	\begin{equation}
		\left\{
		\begin{array}{l}
			\mathop{\min}\limits_{\substack{C\in\mathcal{C}_{f}\\ s.t. |C|=q}} \mathbf{x_{C}} \succ \mathop{\max}\limits_{\substack{C\in \mathcal{C}_{f}\\ s.t. |C|= q+1}}\mathbf{x_{C}},\quad \forall {q}\geq \hat{q}_{f}
			\\
			\mathop{\min}\limits_{\substack{C\in\mathcal{C}_{f}\\ s.t. |C|=q}}\mathbf{x_{C}} \succ \mathop{\max}\limits_{\substack{C\in \mathcal{C}_{f}\\ s.t. |C|= q-1}}\mathbf{x_{C}},\quad \forall {q}\leq \hat{q}_{f}
		\end{array}
		\right.
	\end{equation}
	
	Using the results if Proposition \ref{prop:controlincentives}, we immediately obtain the transformation to be applied to the characteristic function to provide the required incentives, 
	\begin{equation}\label{eq:fearofruinrequirement1}
			\mathop{\max}\limits_{\substack{C\in\mathcal{C}_{f}\\ s.t. |C|=q}}F_{C}\circ\Omega(v)(C)
			<
			\mathop{\min}\limits_{\substack{C\in \mathcal{C}_{f}\\ s.t. |C|= q+1}}F_{C}\circ\Omega(v)(C),\quad \forall {q}\geq \hat{q}_{f}
	\end{equation}
	and
	\begin{equation}
			\mathop{\max}\limits_{\substack{C\in\mathcal{C}_{f}\\ s.t. |C|=q}}F_{C}\circ\Omega(v)(C)
			<
			\mathop{\min}\limits_{\substack{C\in \mathcal{C}_{f}\\ s.t. |C|= q-1}}F_{C}\circ\Omega(v)(C),\quad \forall {q}\leq \hat{q}_{f}
	\end{equation}
	where $F_{C} = \left(\sum\limits_{i\in C}\left(\frac{u_{i}^{'}}{u_{i}}\right)^{-1}\right)^{-1}$.\\
	This concludes the proof.
\end{proof}

\subsection{Proof of Proposition~\ref{prop:ConvBDAA}}

\begin{proof}
After initialization, BDAA is made of two loops. 
The first one is a loop of proposals from users. At each iteration of this outer loop, there is an inner loop of counter-proposals from the APs to the users. We show that these two loops stop after a finite number of iterations. Let's first consider the inner loop. At each iteration of the inner loop, the following events can occur: 
\begin{itemize}
	\item An engaged AP remains engaged. Its dynamic list is left unchanged (in Step 2.f, only the dynamic lists of unengaged APs are updated).
	\item An unengaged AP is now engaged. Its dynamic list is left unchanged (in Step 2.f, only the dynamic lists of unengaged APs are updated).
	\item An unengaged AP remains unengaged. This is the case when some of the users it counter-proposed in Step 2.d have rejected its counter-proposal and either (a) none of them is engaged with another AP, or (b) some of them are engaged. In (a) the dynamic list remains unchanged. In (b) it is strictly decreasing. 
	\item An engaged AP becomes unengaged. This means that some users in the coalition it was engaged to defected. This is only possible if they are engaged in a new coalition with another AP (Step 2.e). These defecting users are thus removed from the AP dynamic list, which is strictly decreasing. 
\end{itemize}
In all cases, all the dynamic lists are weakly decreasing in the sense of inclusion. The inner loop thus converges in a finite number of steps. 

We now consider the outer loop. We immediately have the convergence by finiteness of the number of APs each mobile can propose to and the fact that no mobile can propose more than once to any AP. The algorithm converges in a finite number of steps.
\end{proof}

\subsection{Proof of Proposition~\ref{prop:StabBDAA}}
\begin{proof}
		In \cite{Pyc2012}, it is shown that the stability inducing sharing rules as given in Proposition~\ref{prop:stability} (see Section~\ref{sec:matching}) induce pairwise aligned preferences profiles satisfying the richness condition $R1$ (\cite{Pyc2012}, pp. 334) of the domain of preferences $\boldsymbol{R}$. 
	It is shown (\cite{Pyc2012}, Lemma 3, pp. 349) that if the family of coalitions $\mathcal{C}$, the preferences domain $\boldsymbol{R}$ satisfies $R1$ and all preference profiles are pairwise aligned, then no profile $\boldsymbol{R}$  admits an n-cycle, $n\in\{3,4,\ldots\}$. 
	As a consequence, there exists a core stable structure.
	It is unique if the preferences are strict. 
	Furthermore, in the proof of Proposition 5 (\cite{Pyc2012}, Proposition 5, pp. 359) it is shown that without n-cycles, the coalitions $\{C_{1},\ldots,C_{k}\}$ in the stable structure can be re-indexed as $\{C_{i_{1}},\ldots,C_{i_{k}}\}$  so that $ 
C_{i_{j}}$ is weakly preferred by its members to any coalition of agents in $\mathcal{N}\backslash\{C_{i_{1}},\ldots,C_{i_{j-1}}\}$. 
If the preferences are strict (using a tie-breaking rule in case of indifference), then the members of $C_{i_{j}}$ strictly prefer it.
\\

	We show by induction that these coalitions are formed in BDAA and never blocked by any other coalition once formed.\\
	$\bullet$ The mobiles in $C_{i_{1}}$ propose to the AP in $C_{i_{1}}$ in the first proposal round since it is the most preferred coalition of any agent forming it. 
	The AP in $C_{i_{1}}$  counter-proposes to these mobiles who all accept. 
	The coalition $C_{i_{1}}$ is formed.
	No player of this coalition has any incentive to leave this coalition in a subsequent round. 
	It cannot be blocked.
	\\
	
	$\bullet$ Assume that the coalitions $\{C_{i_{1}}, \ldots, C_{i_{l}}\}$, $l<k$ are formed and are not blocked. 
	We show that $C_{i_{l+1}}$ will be formed and cannot be blocked.  
	Using the previous results, we have that the players in $C_{i_{l+1}}$ prefer it to any other coalition that can be formed with agents in $\mathcal{N}\backslash\{C_{i_{1}},\ldots,C_{i_{l}}\}$. 
	The payoff they receive from any of these coalitions is lower that their payoff in $C_{i_{l+1}}$, thus inferior to the maximum achievable payoff with the AP in $C_{i_{l+1}}$. 
	As a consequence, all the players in $C_{i_{l+1}}$ must have ultimately proposed to the AP in this coalition.
	\\
	At this point, the cumulated and dynamic lists of this AP contain the players in $C_{i_{l+1}}$. 
	It counter-proposes to the coalitions it prefers to $C_{i_{l+1}}$ which contain players in the coalitions $\{C_{i_{1}},\ldots,C_{i_{l}}\}$. 
	These players reject the counter-proposals and are removed from the AP's dynamic list which ultimately counter-proposes to $C_{i_{l+1}}$. 
	Any mobile in $C_{i_{l+1}}$ rejects this counter-proposal to propose to the other APs with maximum achievable payoff higher than its payoff in $C_{i_{l+1}}$. 
	This continues up to the point where no maximum achievable payoff is higher than the payoff in $C_{i_{l+1}}$. 
	At this point, all the players in this coalition accept forming $C_{i_{l+1}}$.
	The coalition $C_{i_{l+1}}$ is formed.
	\\
	No player of this coalition has any incentive to leave this coalition in a subsequent round. 
	This concludes the induction proof.\\
	The unique stable structure $\{C_{i_{1}}, \ldots, C_{i_{k}}\}$ is the output of BDAA.
\end{proof}

\subsection{Proof of Proposition~\ref{prop:CompBDAA}}

\begin{proof}
	First we give an upper bound on the number of proposals emitted by the mobile users, then we give an upper bound on the number of proposals emitted by the APs.
	Finally, we conclude.\\
	In at most $F$ proposals, every mobile user has proposed to all the APs.
	Thus, in at most $F\times W$ proposals, the mobile users have proposed to all the APs.\\
	In at most $W$ counter-proposals, every AP has proposed to all the mobiles.
	Furthermore, every AP counter-proposes at each counter-proposing round.
	Thus, in at most $F\times W\times F$ the APs have emitted all their counter-proposals.
	We obtain that the total number of proposals (both mobile users proposals and APs counter-proposals) cannot exceed $F^{3}\times W^{2}$.
	The complexity of BDAA is  $O(n^{5})$ where $n=\max(F, W)$.
\end{proof}


\section{Interpretation of BDAA in the economic framework} \label{app:BDAA-interpretations}

We show that the proposed matching mechanism (backward deferred acceptance) is particularly instinctive and natural in representing a competitive labor market process.
We particularly focus on a competitive labor market interpretation.

Consider a competitive labor market made of a set $\mathcal{W}$ of workers and a set $\mathcal{F}$ of firms. 
Firms and workers can form coalitions as in the considered association and resource allocation game studied in the paper.
The workers are looking for some jobs and the firms are looking for hiring some job seekers.
We assume backward deferred acceptance mechanism as hiring process.

In Step 1.c, the firms communicate their top perspectives (assumed measured in payoffs) to the job seekers.
Such communication can be understood as taking part in some recruiting campaign.
In Step 2.a, the workers apply to the firms on the basis of the top perspectives and the promises of the recruiting campaigns.
In Step 2.c, the firms consider the set of received applications and propose jobs.
The emitted job offers may not necessarily correspond to the top perspectives of the recruiting campaign.
In fact some conditions may not have been fulfilled (e.g. the set of workers for optimal production).
Each job seeker receives the job proposals and accepts or rejects.
In Step 2.d a job seeker may either accept or reject a counter-proposal.
The acceptance and rejections are received by the firms.
If all the proposals of a firm have been accepted, the workers are engaged (Step 2.e).
Hired workers are declared as is. 
The counter-proposals go on up to the last firm.
 Each firm in the market having proposed jobs removes from its list of candidates those workers having rejected its counter-proposal and been engaged (Step 2.f). 
The candidature of the other job seekers are kept into consideration (Step 2.f).
The firms may tell the job-seekers still in their lists that they are under consideration and that they may receive new offers.
The firms not having recruited yet go on emitting job offers (Step 2.g).  
This process goes on up to the point where no job seeker receives new proposals (Step 2.g).
At this point, the firms may notify the workers that they have received all the proposals.
It is now upon the job seekers to create new opportunities.
The job seekers attempt to send new applications (to their next most preferred firm in terms of top perspectives).
The mechanism goes on up to the point where no job seeker can propose.\\

\section{Example} \label{appendix:Example}

\begin{figure*}[t]
\begin{minipage}[b]{0.22\linewidth}
\centering
		\begin{tikzpicture}[thick,
		  workers/.style={draw, circle, minimum size = 1.5mm, inner sep=0pt, fill = black},
		  firms/.style={draw, circle, minimum size = 1.5mm, inner sep=0pt, fill = black},
		  tasks/.style={draw, circle, minimum size = 5mm, inner sep=0pt}
		]
		
		\begin{scope}[start chain=going below,node distance=1cm]
		\foreach \i in {1,...,2}
		  \node[workers,on chain] (w\i) [label=left: $w_{\i}$] {};
		\end{scope}
		
		\begin{scope}[xshift=2cm,start chain=going below,node distance=1cm]
		  \node[firms,on chain] (f1) [label=right: $f_{1}$] {};
		  \node[firms,on chain] (f2) [label=right: $f_{2}$] {};
		  \node[firms,on chain] (f3) [label=right: $f_{3}$] {};
		\end{scope}

		\draw[<-] (w1) -- (f1) node [midway, above] {\footnotesize 10};
		\draw[<-] (w1) -- (f2) node [midway, above] {\footnotesize 1};
		\draw[<-] (w2) -- (f1) node [midway, below] {\footnotesize 10};
		\draw[<-] (w2) -- (f3) node [midway, below] {\footnotesize 100};
		
		\end{tikzpicture}
	\caption{The APs send the best achievable payoff to the mobiles.}
	\label{fig:BDAAexamplesystem1}
\end{minipage}
\hspace{0.5cm}
\begin{minipage}[b]{0.22\linewidth}
\centering
		\begin{tikzpicture}[thick,
		  workers/.style={draw, circle, minimum size = 1.5mm, inner sep=0pt, fill = black},
		  firms/.style={draw, circle, minimum size = 1.5mm, inner sep=0pt, fill = black},
		  tasks/.style={draw, circle, minimum size = 5mm, inner sep=0pt}
		]
		
		\begin{scope}[start chain=going below,node distance=1cm]
		\foreach \i in {1,...,2}
		  \node[workers,on chain] (w\i) [label=left: $w_{\i}$] {};
		\end{scope}
		
		\begin{scope}[xshift=2cm,start chain=going below,node distance=1cm]
		  \node[firms,on chain] (f1) [label=right: $f_{1}$] {};
		  \node[firms,on chain] (f2) [label=right: $f_{2}$] {};
		  \node[firms,on chain] (f3) [label=right: $f_{3}$] {};
		\end{scope}

		\draw[dashdotted, ->] (w1) -- (f1) node {};
		\draw[dashdotted, ->] (w2) -- (f3) node {};
		
		\end{tikzpicture}
	\caption{The mobiles propose to the APs.}
	\label{fig:BDAAexampleproposals}
\end{minipage}
\begin{minipage}[b]{0.22\linewidth}
\centering
		\begin{tikzpicture}[thick,
		  workers/.style={draw, circle, minimum size = 1.5mm, inner sep=0pt, fill = black},
		  firms/.style={draw, circle, minimum size = 1.5mm, inner sep=0pt, fill = black},
		  tasks/.style={draw, circle, minimum size = 5mm, inner sep=0pt}
		]
		
		\begin{scope}[start chain=going below,node distance=1cm]
		\foreach \i in {1,...,2}
		  \node[workers,on chain] (w\i) [label=left: $w_{\i}$] {};
		\end{scope}
		
		\begin{scope}[xshift=2cm,start chain=going below,node distance=1cm]
		  \node[firms,on chain] (f1) [label=right: $f_{1}$] {};
		  \node[firms,on chain] (f2) [label=right: $f_{2}$] {};
		  \node[firms,on chain] (f3) [label=right: $f_{3}$] {};
		\end{scope}

		\draw[<-, dotted] (w1) -- (f1) node [midway, above] {\footnotesize 0.5};
		\draw[<-, dotted] (w2) -- (f3) node [midway, below] {\footnotesize 100};
		
		\end{tikzpicture}
	\caption{On the basis of the proposals, the APs emit counter-proposals.}
	\label{fig:BDAAexamplecounterproposals}
\end{minipage}
\begin{minipage}[b]{0.22\linewidth}
		\centering
		\begin{tikzpicture}[thick,
		  workers/.style={draw, circle, minimum size = 1.5mm, inner sep=0pt, fill = black},
		  firms/.style={draw, circle, minimum size = 1.5mm, inner sep=0pt, fill = black},
		  tasks/.style={draw, circle, minimum size = 5mm, inner sep=0pt}
		]
		
		\begin{scope}[start chain=going below,node distance=1cm]
		\foreach \i in {1,...,2}
		  \node[workers,on chain] (w\i) [label=left: $w_{\i}$] {};
		\end{scope}
		
		\begin{scope}[xshift=2cm,start chain=going below,node distance=1cm]
		  \node[firms,on chain] (f1) [label=right: $f_{1}$] {};
		  \node[firms,on chain] (f2) [label=right: $f_{2}$] {};
		  \node[firms,on chain] (f3) [label=right: $f_{3}$] {};
		\end{scope}

		\draw[<-, dotted] (w1) -- (f1) node [midway] {$\times$};
		\draw (w2) -- (f3) node [midway, below] {\footnotesize 100};
		
		\end{tikzpicture}
	\caption{The mobiles accept or reject the counter-proposals.}
	\label{fig:BDAAexamplerejects}
\end{minipage}
\end{figure*}

\begin{figure*}[b]
\begin{minipage}[b]{0.22\linewidth}
		\centering
		\begin{tikzpicture}[thick,
		  workers/.style={draw, circle, minimum size = 1.5mm, inner sep=0pt, fill = black},
		  firms/.style={draw, circle, minimum size = 1.5mm, inner sep=0pt, fill = black},
		  tasks/.style={draw, circle, minimum size = 5mm, inner sep=0pt}
		]
		
		\begin{scope}[start chain=going below,node distance=1cm]
		\foreach \i in {1,...,2}
		  \node[workers,on chain] (w\i) [label=left: $w_{\i}$] {};
		\end{scope}
		
		\begin{scope}[xshift=2cm,start chain=going below,node distance=1cm]
		  \node[firms,on chain] (f1) [label=right: $f_{1}$] {};
		  \node[firms,on chain] (f2) [label=right: $f_{2}$] {};
		  \node[firms,on chain] (f3) [label=right: $f_{3}$] {};
		\end{scope}

		\draw[->, dashdotted] (w1) -- (f2) node {};
		\draw (w2) -- (f3) node [midway, below] {\footnotesize 100M};
		
		\end{tikzpicture}
	\caption{Mobile $w_{1}$ proposes to AP $f_{2}$.}
	\label{fig:BDAAexampleproposals2}
\end{minipage}
\hspace{0.5cm}
\begin{minipage}[b]{0.22\linewidth}
		\centering
		\begin{tikzpicture}[thick,
		  workers/.style={draw, circle, minimum size = 1.5mm, inner sep=0pt, fill = black},
		  firms/.style={draw, circle, minimum size = 1.5mm, inner sep=0pt, fill = black},
		  tasks/.style={draw, circle, minimum size = 5mm, inner sep=0pt}
		]
		
		\begin{scope}[start chain=going below,node distance=1cm]
		\foreach \i in {1,...,2}
		  \node[workers,on chain] (w\i) [label=left: $w_{\i}$] {};
		\end{scope}
		
		\begin{scope}[xshift=2cm,start chain=going below,node distance=1cm]
		  \node[firms,on chain] (f1) [label=right: $f_{1}$] {};
		  \node[firms,on chain] (f2) [label=right: $f_{2}$] {};
		  \node[firms,on chain] (f3) [label=right: $f_{3}$] {};
		\end{scope}

		\draw[<-, dotted] (w1) -- (f1) node [midway, above] {\footnotesize 0.5};
		\draw[<-, dotted] (w1) -- (f2) node [midway, above] {\footnotesize 1};
		\draw (w2) -- (f3) node [midway, below] {\footnotesize 100};
		
		\end{tikzpicture}
	\caption{APs $f_{1}$ and $f_{2}$ counter-propose to $w_{1}$.}
	\label{fig:BDAAexamplecounterproposals2}
	\end{minipage}
	\begin{minipage}[b]{0.22\linewidth}
		\centering
		\begin{tikzpicture}[thick,
		  workers/.style={draw, circle, minimum size = 1.5mm, inner sep=0pt, fill = black},
		  firms/.style={draw, circle, minimum size = 1.5mm, inner sep=0pt, fill = black},
		  tasks/.style={draw, circle, minimum size = 5mm, inner sep=0pt}
		]
		
		\begin{scope}[start chain=going below,node distance=1cm]
		\foreach \i in {1,...,2}
		  \node[workers,on chain] (w\i) [label=left: $w_{\i}$] {};
		\end{scope}
		
		\begin{scope}[xshift=2cm,start chain=going below,node distance=1cm]
		  \node[firms,on chain] (f1) [label=right: $f_{1}$] {};
		  \node[firms,on chain] (f2) [label=right: $f_{2}$] {};
		  \node[firms,on chain] (f3) [label=right: $f_{3}$] {};
		\end{scope}

		\draw[<-, dotted] (w1) -- (f1) node [midway] {$\times$};
		\draw[<-, dotted] (w1) -- (f2) node [midway] {$\circ$};
		\draw (w2) -- (f3) node [midway, below] {\footnotesize 100M};
		
		\end{tikzpicture}
	\caption{Mobile $w_{1}$ accept $f_{2}$ counter-proposal and rejects $f_{1}$' one.}
	\label{fig:BDAAexampleacceptance}
	\end{minipage}
	\hspace{0.5cm}
	\begin{minipage}[b]{0.22\linewidth}
		\centering
		\begin{tikzpicture}[thick,
		  workers/.style={draw, circle, minimum size = 1.5mm, inner sep=0pt, fill = black},
		  firms/.style={draw, circle, minimum size = 1.5mm, inner sep=0pt, fill = black},
		  tasks/.style={draw, circle, minimum size = 5mm, inner sep=0pt}
		]
		
		\begin{scope}[start chain=going below,node distance=1cm]
		\foreach \i in {1,...,2}
		  \node[workers,on chain] (w\i) [label=left: $w_{\i}$] {};
		\end{scope}
		
		\begin{scope}[xshift=2cm,start chain=going below,node distance=1cm]
		  \node[firms,on chain] (f1) [label=right: $f_{1}$] {};
		  \node[firms,on chain] (f2) [label=right: $f_{2}$] {};
		  \node[firms,on chain] (f3) [label=right: $f_{3}$] {};
		\end{scope}

		\draw (w1) -- (f2) node [midway, above] {\footnotesize 1};
		\draw (w2) -- (f3) node [midway, below] {\footnotesize 100};
		
		\end{tikzpicture}
	\caption{The stable matching.}
	\label{fig:BDAAexamplematching}
	\end{minipage}	
\end{figure*}

In Figure (\ref{fig:BDAAexamplesystem1}) we show of a first example of application of the BDAA.
The solid arrows show the best achievable payoffs that the mobiles can obtain with each AP.
In dash-dotted we show the proposals of the mobiles to the APs, in dotted the counter-proposals of the APs to the mobiles and in plain the engagement. The cross show a reject and the circles an acceptance.\\

The APs send the maximum achievable payoff to the players (see Figure  \ref{fig:BDAAexamplesystem1}).
 AP $f_{1}$ sends 10 to $w_{1}$ and $w_{2}$ for the coalition $\{f_{1}; w_{1}, w_{2}\}$.
For $\{f_{1}; w_{1}\}$ or $\{f_{1}; w_{2}\}$ the payoff is 0.5 due to the control of the incentives\footnote{In this example the control provides the mobiles $w_{1}$ and $w_{2}$ the incentives for the set of coalitions of size $3$ w.r.t. those of size $2$ in the set of coalitions that the players can form with AP $f_{1}$}.
AP $f_{2}$ sends 1 for $\{f_{2};w_{1}\}$.
AP $f_{3}$ sends 100 for $\{f_{3};w_{2}\}$.\\

The induced mobiles' preferences are thus, $f_{1}\succ_{w_{1}}f_{2}$ and $f_{3}\succ_{w_{2}}f_{1}$.
In the first proposing round, mobile $w_{1}$ propose to the AP $f_{1}$ and mobile $w_{2}$ propose to the AP $f_{3}$ (see Figure \ref{fig:BDAAexampleproposals}).
The APs update their cumulated list $L(f_{1}) = \{w_{1}\}$, $L(f_{2}) = \emptyset$, $L(f_{3}) = \{w_{3}\}$ and their dynamic list $L^{*}(f_{1}) = \{w_{1}\}$, $L^{*}(f_{2}) = \emptyset$, $L^{*}(f_{3}) = \{w_{3}\}$.
The APs emit the following counter proposals: $(f_{1}, w_{1}, 0.5)$, $(f_{3}, w_{2}, 100)$ (see Figure ~\ref{fig:BDAAexamplecounterproposals}).
The mobiles receive the counter-proposals.
Mobile $w_{2}$ accepts the counter-proposal of $f_{3}$ since none of its maximum achievable payoff gives more than the counter-proposal (shown in solid line in Figure~\ref{fig:BDAAexamplerejects}). Mobile $w_{1}$ does not accept the counter-proposal of $f_{2}$ since the maximum achievable payoff received from $f_{2}$ is higher than the current counter-proposal from $f_{1}$ (shown by the cross in Figure~\ref{fig:BDAAexamplerejects}).\\

The dynamic list of the APs are updated, $L^{*}(f_{1}) = \{w_{1}\}$ and $L^{*}(f_{2}) = \emptyset$.
The algorithm enters the second counter-proposing round.
AP $f_{1}$ emit the counter-proposal $(f_{1}, w_{1}, 0.5)$ which is the same as in the previous counter-proposing round.
Mobile $w_{1}$ still rejects.
The dynamic list of the unengaged APs are updated, $L^{*}(f_{1}) = \{w_{1}\}$ and $L^{*}(f_{2}) = \emptyset$.\\

The dynamic lists of the unengaged APs have not changed. The counter-proposing loop stops and the algorithm enters the outer loop for proposals.
The second proposing round starts.
The mobile $w_{1}$ is the only unengaged mobile.
Its next most preferred AP is $f_{2}$ (with best achievable payoff 1).
The mobile $w_{1}$ propose to $f_{2}$ (see Figure \ref{fig:BDAAexampleproposals2}).
The APs' lists are updated, $L(f_{1}) = L^{*}(f_{1}) = \{w_{1}\}$, $L(f_{2}) = L^{*}(f_{2}) = \{w_{1}\}$ and $L(f_{3}) = L^{*}(f_{3}) = \{w_{2}\}$.
The two unengaged APs $f_{1}$ and $f_{2}$ counter-propose (see Figure~\ref{fig:BDAAexamplecounterproposals2}).
They emit the following counter proposals: $(f_{1}, w_{1}, 0.5)$, $(f_{2}, w_{1}, 1)$.
Mobile $m_{1}$ receives the two counter-proposals, accepts the one of $f_{2}$ and rejects the one of $f_{1}$.
It is engaged to $f_{2}$ (see Figure~\ref{fig:BDAAexampleacceptance}).
The dynamic list of the unengaged AP $f_{1}$ is updated, $L^{*}(f_{1}) = \emptyset$.
The only unengaged AP cannot propose.
There are no more unengaged mobiles.\\
The algorithm stops. The final stable matching is shown in Figure \ref{fig:BDAAexamplematching}.

\end{document}